\documentclass[11pt]{article}

\usepackage{amsmath,amsthm,amssymb,amscd}
\usepackage{ascmac}
\usepackage{euscript}
\usepackage{enumerate}

\usepackage[utf8]{inputenc} 
\usepackage[T1]{fontenc}    
\usepackage{hyperref}       
\usepackage{booktabs}       
\usepackage{amsfonts}       
\usepackage{bm,color,geometry,mathrsfs}
\usepackage[pdftex]{graphicx}
\def\dim{\mathop{\mathrm{dim}}\nolimits}

\def\Im{\mathop{\mathrm{Im}}\nolimits}
\def\Re{\mathop{\mathrm{Re}}\nolimits}

\def\Span{\mathop{\mathrm{Span}}\nolimits}

\def\i{\mathrm{i}}

\newcommand{\mb}[1]{{\mathbf{#1}}}

\def\C{\mathbb{C}}

\def\R{\mathbb{R}}

\def\Z{\mathbb{Z}}
\def\CA{\mathcal{A}}
\def\CC{\mathcal{C}}

\def\CF{\mathcal{F}}

\def\x{\mathrm{x}}
\def\y{\mathrm{y}}
\def\ww{w}
\def\d{\partial}
\def\lap{\mathcal{L}}

\numberwithin{equation}{section}
\newtheorem{thm}[equation]{Theorem}

\newtheorem{exa}[equation]{Example}
\newtheorem{prop}[equation]{Proposition}
\newtheorem{cor}[equation]{Corollary}
\newtheorem{lem}[equation]{Lemma}

\newtheorem{rem}[equation]{Remark}

\newtheorem{algo}[equation]{Algorithm}

\numberwithin{table}{section}
\numberwithin{figure}{section}


\title{Graph recovery from graph wave equation}

\author{%
  Yuuya Takayama \\
 Nikon Corporation, Japan \\
 \texttt{yuuya.takayama@nikon.com} \\
}

\date{\empty}

\begin{document}
\maketitle

\begin{abstract}

We propose a method by which to recover an underlying graph from a set of multivariate wave signals that is discretely sampled from a solution of the graph wave equation.
Herein, the graph wave equation is defined with the graph Laplacian, and its solution is explicitly given as a mode expansion of the Laplacian eigenvalues and eigenfunctions.
For graph recovery, our idea is to extract modes corresponding to the square root of the eigenvalues from the discrete wave signals using the DMD method, and then to reconstruct the graph (Laplacian) from the eigenfunctions obtained as amplitudes of the modes.
Moreover, in order to estimate modes more precisely, we modify the DMD method under an assumption that only stationary modes exist, because graph wave functions always satisfy this assumption.

In conclusion, we demonstrate the proposed method on the wave signals over a path graph.
Since our graph recovery procedure can be applied to non-wave signals, we also check its performance on human joint sensor time-series data.

\end{abstract}

\section{Introduction}

Graphs have been researched over a century in mathematics and have recently been important in science and application fields, such as physics, chemistry, finance, signal processing, and machine learning.
Once a graph is given, mathematical studies can give powerful and useful results about its Laplacian eigenvalues, its connectivity, its hitting time, and so on.
However, for practical data, the graph structure is sometimes unknown.
Hence, it is a fundamental problem to reconstruct the underlying graph from a given dataset.
As typical examples of its solution, Gaussian graphical models (GGMs) are well known in machine learning; these models are derived from the assumption that data are generated from a multivariate Gaussian distribution \cite{uhler2017gaussian}.
In the present paper, alternatively, we assume that multivariate time-series data are generated from the \emph{graph wave equation} and then provide a procedure by which to recover the weighted graph structure over the variables from the observed dataset.
Moreover, based on this procedure, we propose an algorithm to construct a reasonable graph from any multivariate time-series data without assuming the data to follow the graph wave equation.

The graph wave equation was first studied by Friedman and Tillich
\cite{friedman2004wave} as a discrete analogue of the continuous wave equation.
Here, its general solutions, i.e. \emph{graph wave functions}, are given as
\begin{align*}
U(\x,t) = \sum_{k} \left\{ a_k \cos{\sqrt{\rho_k} t} + \frac{b_k}{\sqrt{\rho_k}} \sin{\sqrt{\rho_k} t} \right\}v_k(\x),
\end{align*}
by the eigenfunctions $\{v_k\}$ of the graph Laplacian, where $\x$ indicates a node of a graph.
In fact, a graph wave function for a path graph is viewed as a spatial discretization of a conventional wave function over an interval.
In the present paper, as a discrete multivariate time-series model, we use a set of observed values of a graph wave function at integer points $t\in\Z$, namely, $\{U(\x,t)\mid \forall \x, t \in \Z \}$.
Note that this discretization in the time direction does not lose any information of the waves if the Nyquist condition $\sqrt{\rho_k} < \pi$ is satisfied for any $k$.
When the underlying graph is unknown, so are the frequencies $\{\sqrt{\rho_k}\}$.
In contrast, it is easier to estimate the eigenfunctions $\{v_k\}$ when the frequencies are known because of the curve-fitting method, or in other words, the least squares method \cite{forsythe1957generation} \cite[\S 5]{li2013time}.
Hence, in order to recover the graph, it is necessary to estimate the frequencies from observed values of a graph wave function.

To this end, we use the \emph{dynamic mode decomposition} (DMD) proposed by Schmid \cite{schmid2010dynamic} and studied by Tu et al. \cite{tu2013dynamic}.
The essence of DMD is to describe the $(T+1)$-th observed set in terms of a linear sum of the previous $T$ observed sets
\begin{align*}
U(\x,T+1) = c_1U(\x,T) + c_2U(\x, T-1) + \cdots + c_TU(\x,1)
\end{align*}
by $c_j \in \R$.
Then, the modes $\{z_j\}$ are calculated as the solutions of a polynomial equation $Z^T - c_1Z^{T-1} - \cdots - c_{T-1}Z - c_T = 0$, just like \emph{Prony's method} \cite{prony1795essai}, but are shared in all (spatial) variables $\{U(\x, \cdot)\}$.
Moreover, the mode $z_j$ is immediately related to the frequency $\sqrt{\rho_k}$ as $z_j = e^{\pm\i\sqrt{\rho_k} t}$ for some $j$ and $k$.
Here, it is worth emphasizing that the mode set $\{z_j\}$ satisfies a stationary condition $|z_j| = 1$ and a conjugate condition $\bar{z}_j = z_{l}$ (for some $l$).
Hence, the coefficients $\{c_j\}$ satisfy a symmetric condition \cite{li1993asymptotic}.
In fact, this idea works well to calculate the frequencies robustly against numerical error from fewer observation points along the time axis.
As a result, this modified DMD method enables us to extract common frequencies more efficiently from multivariate time-series data, although other approaches may be applicable to the frequency estimation for the time-series data of each variable, such as \cite{li2013time, plonka2014prony, stampfer2020generalized}, 

From the calculated frequencies, we can estimate the graph Laplacian and the graph weights, in addition to the eigenfunctions $\{v_k\}$, as explained above.
To be more precise, the graph weights are given as
\begin{align*}
\ww_{\x\y} = -\sum_{k}v_k(\x)\rho_kv_k(\y).
\end{align*}
For any multivariate time-series data, we extend this form by replacing $\{v_k\}$ with estimated eigenfunction-like functions.
Here, the constructed frequency-based weights have the following remarkable properties:
\begin{itemize}
\item the weights are independent of the mean of each variable,
\item when two variables $U(\x,\cdot)$ and $U(\y,\cdot)$ have no common frequency, then $\ww_{\x\y}=0$, and
\item when two variables $U(\x,\cdot)$ and $U(\y,\cdot)$ are anti-synchronized, then $\ww_{\x\y}>0$.
\end{itemize}

In contrast to the frequency-based graph, Gaussian graphical models (GGMs) define graph weights by the correlation between two variables after conditioning on all other datasets \cite[\S 8]{uhler2017gaussian}
and are widely applicable to biology \cite{wang2016fastggm}, psychology \cite{epskamp2018gaussian}, and finance \cite{giudici2016graphical}.
Another perspective suggests defining distances for time series, such as dynamical time warping (DTW) \cite{berndt1994using} and feature based distances \cite{faloutsos1994fast, popivanov2002similarity}.
These distances define graph weights through distance kernels for the sake of various time-series analysis or machine learning tasks \cite{abanda2019review}. 
In particular, for the human joint time-series data we deal with, a graph structure is often constructed based on the Joint Relative Distance, defined by the pair-wise Euclidean distances between two 3D (or 2D) joint locations taken by a physical sensor.
Namely, graph weights are given by its variance \cite{tang2012retrieval}, its relative variance \cite{li2017graph}, or its DTW \cite{ahmed2015dtw}.
Compared with these studies, the proposed method has advantages in that it can directly extract information of temporal order, especially, the frequency of time-series data, and moreover requires a smaller amount of data.

The remainder of the present paper is organized as follows.
In \S \ref{sec:mode_extract}, we review the way to extract modes from equispaced sampled (multivariate) time-series data, known as Prony's method, and the DMD.
We also explain how to impose a stationary mode condition on these methods.
In \S \ref{sec:graphwave}, we first recall the definitions of the graph Laplacian and the graph wave equation.
Next, we see that a solution to the graph wave equation, a graph wave function, for a path graph naturally corresponds to a continuous wave over an interval, which is used for a demonstration in \S \ref{sec:experiment}.
Then, we solve the graph recovery problem for an observed dataset of graph wave functions using the modified DMD method.
This procedure is summarized and extended in \S \ref{sec:experiment}.
Finally, we determine its performance through three examples: graph wave signals over a path graph, discretely sampled continuous wave signals over an interval, and human joint sensor time-series data.

\section{Mode extraction techniques}
\label{sec:mode_extract}
In this section, we review mode extraction methods from the observed multivariate time-series dataset.
Then, we explain how to modify these methods for a set of modes satisfying special conditions in \S \ref{subsec:special_mode}.
In \S \ref{subsec:mode_expansion}, we also mention amplitude estimation methods, which are combined into mode expansion techniques.
Finally, we briefly explain examples of the mode expansion results in \S \ref{subsec:example}.

\subsection{Mode extraction problem setting}
Let $V$ be a finite set $\{ 1, 2, \cdots, n\}$ and let $\CA_V$ denote the set of functions $\{f\colon V \rightarrow \R \}$, which is an $\R$-vector space.
When $n=1$, $\CA_V$ is identified with $\R$.
We define two constant functions $0_V(\x) := 0$ and $1_V(\x) := 1$ for any $\x \in V$.

In addition, let $\{z_1, \cdots, z_{M}\}$ be roots of a polynomial equation
\begin{align}
Z^M - c_{1}Z^{M-1} - c_{2}Z^{M-2} - \cdots - c_{M-1}Z - c_M =0.
\label{eq:polynomial}
\end{align}
Needless to say, each coefficient is represented by $c_j = (-1)^{j+1}e_j(z_1,\cdots,z_{M})$, where $e_j$ denotes the $j$-th elementary symmetric polynomial.
We assume that the roots $\{z_j\}$ are distinct, and then define a time-series model $F$ as the time evolution of the roots
\begin{align}
F(\x,t) := \sum_{j=1}^{M}\alpha_j(\x)z_j^t, \text{ for }\x \in V, t \in \Z,
\label{eq:general_model}
\end{align}
where $\alpha_j\in \CA_V\otimes\C\setminus 0_V$.
We can define $F$ for $t \in \R$ by fixing $\arg z_j$, although this definition is not used herein.
In this context, we just refer to the roots $\{z_j\}$ and coefficients $\{\alpha_j\}$ as the modes and amplitudes, respectively, of $F$.
Now, we assume that the amplitudes satisfy a non-degenerate condition
\begin{align}
\dim_{\C}\Span[\alpha_1, \cdots, \alpha_M] = \min(M,n).
\label{eq:non-degenerate}
\end{align}
Note that this condition is always satisfied by taking subset $V' \subset V$ and considering all of the above over $V'$ if necessary.
Then, because of \eqref{eq:polynomial}, time-series model $F$ satisfies
\begin{align}
F(\x,t+M) - c_1F(\x,t+M-1) - \cdots - c_{M-1}F(\x,t+1)-c_MF(\x,t) = 0
\label{eq:dmd_eq}
\end{align}
for any $\x \in V$ and any $t \in \Z$.
This relation is regarded as a linear equation of the coefficients $\{c_j\}$.
Hence the modes $\{z_j\}$ of $F$ can be determined from temporal observed sets $\{F(\cdot,t)\}$.
\begin{prop}
\label{prop:prony}
Set $L=\max(M-n,0)+1$.
The mode set $\{z_j\}$ in the model $F$ \eqref{eq:general_model} is determined from observed values at $L+M$ integer points $\{F(\cdot,t) \in \CA_V\otimes\C \mid t = 1, 2, \cdots, L+M \}$ if condition \eqref{eq:non-degenerate} is satisfied.
\end{prop}
This result is viewed as \emph{Prony's method} \cite{prony1795essai} when $n=1$, and the \emph{dynamic mode decomposition} \cite{schmid2010dynamic} when $n\geq M$.

\begin{proof}
Let us consider the following linear equation according to \eqref{eq:dmd_eq};
\begin{align}
\left[\begin{smallmatrix}
F(1,1) & F(1,2) & \cdots & F(1,M) \\
\vdots & \vdots & & \vdots \\
F(n,1) & F(n,2) & \cdots & F(n,M) \\
F(1,2) & F(1,3) & \cdots & F(1,M+1) \\
\vdots & \vdots & & \vdots \\
F(n,2) & F(n,3) & \cdots & F(n,M+1) \\
\vdots & \vdots & & \vdots \\
F(1,L) & F(1,L+1) & \cdots & F(1,L+M-1) \\
\vdots & \vdots & & \vdots \\
F(n,L) & F(n,L+1) & \cdots & F(n,L+M-1) \\
\end{smallmatrix}\right]
\left[\begin{smallmatrix}
c_{M} \\ c_{M-1} \\ \vdots \\ c_1
\end{smallmatrix}\right]
=\left[\begin{smallmatrix}
F(1,M+1) \\
\vdots \\
F(n,M+1) \\
F(1,M+2) \\
\vdots \\
F(n,M+2) \\
\vdots \\
F(1,L+M) \\
\vdots \\
F(n,L+M)
\end{smallmatrix}\right].
\label{eq:dmd_matrix}
\end{align}
Here, our goal is to show that the $(nL, M)$ matrix on the left has rank $M$.
It is rewritten as
\begin{align*}
\left[\begin{smallmatrix}
\alpha_1(1) & \alpha_2(1) & \cdots & \alpha_M(1) \\
\vdots & \vdots & & \vdots \\
\alpha_1(n) & \alpha_2(n) & \cdots & \alpha_M(n) \\
\alpha_1(1)z_1 & \alpha_2(1)z_2 & \cdots & \alpha_M(1)z_M \\
\vdots & \vdots & & \vdots \\
\alpha_1(n)z_1 & \alpha_2(n)z_2 & \cdots & \alpha_M(n)z_M \\
\vdots & \vdots & & \vdots \\
\alpha_1(1)z_1^{L-1} & \alpha_2(1)z_2^{L-1} & \cdots & \alpha_M(1)z_M^{L-1} \\
\vdots & \vdots & & \vdots \\
\alpha_1(n)z_1^{L-1} & \alpha_2(n)z_2^{L-1} & \cdots & \alpha_M(n)z_M^{L-1}
\end{smallmatrix}\right]
\left[\begin{smallmatrix}
z_1 & z_1^2 & \cdots & z_1^M \\
z_2 & z_2^2 & \cdots & z_2^M \\
\vdots & \vdots & & \vdots \\
z_M & z_M^2 & \cdots & z_M^M
\end{smallmatrix}\right]
\end{align*}
by applying the definition of $F$ \eqref{eq:general_model}.
Since $\{z_j\}$ are assumed to be distinct, the Vandermonde matrix on the right has full rank.
We claim that the matrix on the left also has rank $M$.
When $n\geq M$, that is $L=1$, this follows from condition \eqref{eq:non-degenerate}.
Suppose $n < M$, that is $L = M - n + 1$.
Take $M$-dimensional row vectors
\begin{align*}
\mb{v}_{j,k} = \left(\alpha_1(j)z_1^k, \alpha_2(j)z_2^k, \cdots, \alpha_M(j)z_M^k\right)
\end{align*}
for $j = 1, 2, \cdots, n$ and $k = 0, 1, \cdots, L-1$.
Non-degenerate condition \eqref{eq:non-degenerate} means $n$ vectors $\{\mb{v}_{j,0}\}$ are linearly independent.
If necessary, retake $\mb{v}_{1,0}$ by a linear combination of $\{\mb{v}_{j,0}\}$ so that all entries of $\mb{v}_{1,0}$ are non-zero, which is possible because $\alpha_j\in \CA_V\otimes\C\setminus 0_V$.
Hence, we first obtain
\begin{align*}
\dim_\C\Span[\mb{v}_{1,0}, \mb{v}_{1,1}, \cdots, \mb{v}_{1,L-1}] = L
\end{align*}
because
\begin{align*}
\left[\begin{smallmatrix}
\alpha_1(1) & \alpha_2(1) & \cdots & \alpha_M(1) \\
\alpha_1(1)z_1 & \alpha_2(1)z_2 & \cdots & \alpha_M(1)z_M \\
\vdots & \vdots & & \vdots \\
\alpha_1(1)z_1^{L-1} & \alpha_2(1)z_2^{L-1} & \cdots & \alpha_M(1)z_M^{L-1}
\end{smallmatrix}\right]
=\left[\begin{smallmatrix}
1   & 1   & \cdots & 1   \\ 
z_1 & z_2 & \cdots & z_M \\
\vdots & \vdots & & \vdots \\
z_1^{L-1} & z_2^{L-1} & \cdots & z_M^{L-1}
\end{smallmatrix}\right]
\left[\begin{smallmatrix}
\alpha_1(1)   & 0   & \cdots & 0   \\ 
0 & \alpha_2(1) & & 0 \\
\vdots & & \ddots & \\
0 & 0 & \cdots & \alpha_M(1)
\end{smallmatrix}\right]
\end{align*}
Next, if $\mb{v}_{2,0} \in \Span[\mb{v}_{1,0}, \mb{v}_{1,1}, \cdots, \mb{v}_{1,L-1}]$, then we can describe $\mb{v}_{2,0} = \sum_{l=0}^{L-1} \beta_l\mb{v}_{1,l}$ for some $\beta_l$.
This also indicates $\mb{v}_{2,L-1} = \sum_{l=0}^{L-1} \beta_l\mb{v}_{1,L+l-1}$.
Hence, we obtain
\begin{align*}
\dim_\C\Span[\mb{v}_{1,0}, \mb{v}_{1,1}, \cdots, \mb{v}_{1,L-1}, \mb{v}_{2,0}, \mb{v}_{2,1}, \cdots, \mb{v}_{2,L-1}] \geq L+1,
\end{align*}
because condition \eqref{eq:non-degenerate} means $\mb{v}_{2,L-1} \neq \beta_0\mb{v}_{1,L-1}$.
Otherwise, $\Span[\mb{v}_{1,0}, \mb{v}_{1,1}, \cdots, \mb{v}_{1,L-1}, \mb{v}_{2,0}]$ becomes an $(L+1)$-dimensional vector space again.
By induction, we can conclude
\begin{align*}
\dim_\C\Span[\mb{v}_{1,0}, \cdots, \mb{v}_{1,L-1}, \mb{v}_{2,0}, \cdots, \mb{v}_{2,L-1}, \cdots, \mb{v}_{n,0}, \cdots, \mb{v}_{n,L-1}] \geq L+n-1 = M.
\end{align*}
Therefore, the coefficients $\{c_j\}$ are obtained as a unique solution of the linear equation \eqref{eq:dmd_matrix}, and the modes $\{z_j\}$ are determined as solutions of \eqref{eq:polynomial}.
\end{proof}
\begin{rem}
In practice, it is often the case that the number of modes is also unknown.
In that case, one can compute the number of modes as the rank of the matrix on the left of \eqref{eq:dmd_matrix} for large $M$ and $L$.
\end{rem}

\subsection{Special conditions for mode set}
\label{subsec:special_mode}
Now, we consider a special case of \eqref{eq:general_model}.
In practice, the time-series model $F$ often contains a constant term, which is represented as a mode $Z=1$.
In addition, it is usual to suppose that other modes are stationary and $F$ is $\R$-valued.
\begin{figure}[t]
\begin{center}
\includegraphics[width=4cm]{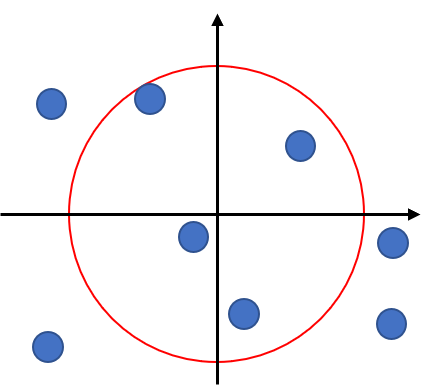}
\
\includegraphics[width=4cm]{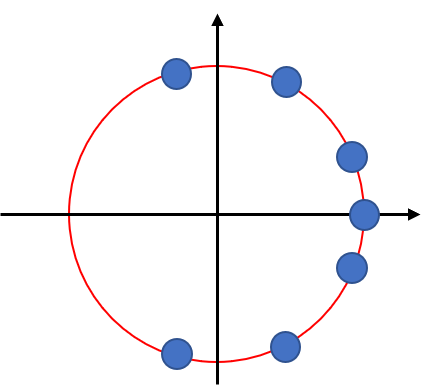}
\
\includegraphics[width=4cm]{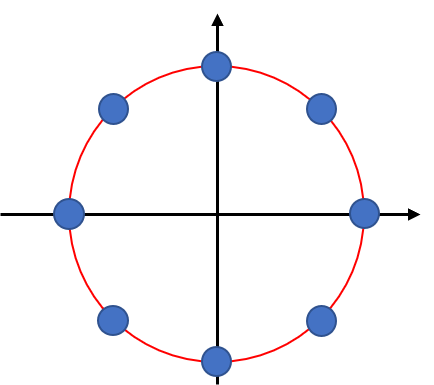}
\caption{[left] General mode set in \eqref{eq:general_model}.
[right] DFT mode set, lying on the unit circle at equal intervals.
[center] Particular mode set in \eqref{eq:real_const_model}, lying on the unit circle vertically symmetrically, including $Z=1$, like DFT modes, but not equispaced.}
\label{fig:modes}
\end{center}
\end{figure}
These conditions indicate a particular representation
\begin{align}
F(\x,t) = \alpha_0(\x) + \sum_{j=1}^{N}\left\{ \alpha_j(\x)z_j^t + \bar{\alpha}_j(\x)\bar{z}_j^t\right\}, \text{ for }t \in \Z,
\label{eq:real_const_model}
\end{align}
where $\alpha_0 \in \CA_V\setminus 0_V$, $\alpha_j \in \CA_V\otimes \C\setminus 0_V$ and $\{z_j \in \C \mid \Im z_j > 0, |z_j|=1$ for $1 \leq j \leq N$ and are distinct from each other$\}$.
This form is also viewed as a general form of the inverse discrete Fourier transform (DFT), the modes of which are given by $\{z_j=\exp(2\pi\i j /N)\}$.
See the difference of mode sets in Figure \ref{fig:modes}.
As above, we assume that the amplitudes satisfy the non-degenerate condition
\begin{align}
\dim_{\C}\Span[\alpha_0, \alpha_1, \cdots, \alpha_N] = \min(N+1,n).
\label{eq:non-degenerate2}
\end{align}
Then, as a special case of \eqref{eq:dmd_eq}, we have the following lemma;
\begin{lem}
\label{lem:real_const_dmd}
Take $d_j = (-1)^{j+1}e_j(z_1,\cdots,z_{N}, \bar{z}_1, \cdots, \bar{z}_{N})$ for $1 \leq j \leq 2N$.
Then, $G(\x,t;k):=F(\x,t+N+1+k)-F(\x,t+N+k)+F(\x,t+N+1-k)-F(\x,t+N-k) \ (0 \leq k \leq N)$ satisfies
\begin{align}
G(\x,t;N) - d_1G(\x,t;N-1) - \cdots - d_{N-1}G(\x,t;1) - \frac{1}{2}d_{N} G(\x,t;0) = 0.
\label{eq:real_const_dmd}
\end{align}
\end{lem}
\begin{proof}
For simplicity, set $z_0 = 1$.
Then, all $Z = z_0, z_1, \cdots, z_{N}, \bar{z}_{1}, \cdots, \bar{z}_{N}$ satisfy
\begin{align*}
(Z-1)\left(Z^{2N} - d_1Z^{2N-1} - d_2Z^{2N-2} - \cdots - d_{2N-1}Z - d_{2N}\right) = 0.
\end{align*}
Now we have $d_{2N} = -1$ and $d_{2N-j} = d_j$ because of the conditions for the mode set $\{z_j\}$, as used in \cite{li1993asymptotic}.
Hence, the above equation can be rewritten as
\begin{align*}
Z^{2N+1} - Z^{2N}+Z-1&-d_1(Z^{2N} - Z^{2N-1}+Z^2-Z) - \cdots \\
& - d_{N-1}(Z^{N+2}-Z^{N+1}+Z^{N}-Z^{N-1}) - d_{N}(Z^{N+1}-Z^{N}) = 0.
\end{align*}
Therefore, by replacing $Z^k$ by $F(\x,t+k)$, we obtain the assertion.
\end{proof}
\begin{rem}
\label{rem:stationary_mode}
For any mode set $\{z_j \in \C \mid |z_j|=1, 1 \leq j \leq N\}$, the coefficients $\{d_j\}$ given in Lemma \ref{lem:real_const_dmd} satisfy $d_{2N} = -1$ and $d_{2N-j} = d_j$, as above.
However, the converse is not always true.
Namely, for the polynomial equation
\begin{align*}
Z^{2N} - d_{1}Z^{2N-1} - \cdots - d_{N-1}Z^{N+1} - d_NZ^{N} - d_{N-1}Z^{N-1} \cdots - d_{1}Z +1 =0,
\end{align*}
when $Z=z$ is a root, $\bar{z}$ and $1/z$ become roots as well, where these two roots do not always coincide, as might be expected.
For example, we can give a counterexample $(Z-2\i)(Z+2\i)(Z-\frac{1}{2}\i)(Z+\frac{1}{2}\i) = Z^4 + \frac{17}{4}Z^2 + 1$ for $N=2$.
\end{rem}
The above lemma helps us to determine the modes from fewer observation points than Proposition \ref{prop:prony}, which needs $L+M = 2N+2+\max(2N+1-n,0)$ observation points for model \eqref{eq:real_const_dmd}.
\begin{prop}
\label{prop:real_const_prony}
Set $L=\max(N-n,0)+1$.
The mode set in model F of \eqref{eq:real_const_model} is determined from observed values at $L+2N+1$ integer points $\{F(\cdot,t) \in \CA_V \mid t = 1, 2, \cdots, L+2N+1 \}$ if condition \eqref{eq:non-degenerate2} is satisfied.
\end{prop}
\begin{proof}
For linear equation \eqref{eq:dmd_matrix} for the coefficients $\{c_j\}$, take $M = 2N+1$ and $L$ as asserted.
Then, the $(nL,2N+1)$ matrix has rank $N$ because of non-degenerate condition \eqref{eq:non-degenerate2}.
By using the argument in the proof of Lemma \ref{lem:real_const_dmd}, we can transform that linear equation into the following equation for the coefficients $\{d_j\}$:
\begin{align}
\left[\begin{smallmatrix}
\frac{1}{2}G(1,1;0) & G(1,1;1) & \cdots & G(1,1;N-1) \\
\vdots & \vdots & & \vdots \\
\frac{1}{2}G(n,1;0) & G(n,1;1) & \cdots & G(n,1;N-1) \\
\frac{1}{2}G(1,2;0) & G(1,2;1) & \cdots & G(1,2;N-1) \\
\vdots & \vdots & & \vdots \\
\frac{1}{2}G(n,2;0) & G(n,2;1) & \cdots & G(n,2;N-1) \\
\vdots & \vdots & & \vdots \\
\frac{1}{2}G(1,L;0) & G(1,L;1) & \cdots & G(1,L;N-1) \\
\vdots & \vdots & & \vdots \\
\frac{1}{2}G(n,L;0) & G(n,L;1) & \cdots & G(n,L;N-1)
\end{smallmatrix}\right]
\left[\begin{smallmatrix}
d_{N} \\ d_{N-1} \\ \vdots \\ d_1
\end{smallmatrix}\right]
=\left[\begin{smallmatrix}
G(1,1;N) \\
 \vdots \\
G(n,1;N) \\
G(1,2;N) \\
 \vdots \\
G(n,2;N) \\
 \vdots \\
G(1,L;N) \\
 \vdots \\
G(n,L;N) 
\end{smallmatrix}\right].
\label{eq:real_const_dmd_matrix}
\end{align}
Here, the $(nL,N)$ matrix on the left still has rank $N$, and then the coefficients $\{d_j\}$ and the modes $\{z_j\}$ are determined.
\end{proof}

\subsection{Amplitude estimation}
\label{subsec:mode_expansion}
Once the modes $\{z_j\}$ in the general model \eqref{eq:general_model} or the particular model \eqref{eq:real_const_model} are obtained, the corresponding amplitudes $\{\alpha_j\}$ are also calculated by the curve fitting method.
For model \eqref{eq:general_model}, we have
\begin{align}
\left[\begin{smallmatrix}
F(1,1) & F(1,2) & \cdots & F(1,T) \\
\vdots & \vdots & & \vdots \\
F(n,1) & F(n,2) & \cdots & F(n,T)
\end{smallmatrix}\right]
=\left[\begin{smallmatrix}
\alpha_1(1) & \alpha_2(1) & \cdots & \alpha_M(1) \\
\vdots & \vdots & & \vdots \\
\alpha_1(n) & \alpha_2(n) & \cdots & \alpha_M(n)
\end{smallmatrix}\right]
\left[\begin{smallmatrix}
z_1 & z_1^2 & \cdots & z_1^T \\
z_2 & z_2^2 & \cdots & z_2^T \\
\vdots & \vdots & & \vdots \\
z_M & z_M^2 & \cdots & z_M^T
\end{smallmatrix}\right],
\label{eq:amp}
\end{align}
for any $T\in\Z_{>0}$.
Since the Vandermonde matrix on the right has rank $M$ when $T \geq M$, observed values at $M$ integer points $\{F(\cdot,t)\in\CA_V\otimes\C\mid t = 1, 2, \cdots, M\}$ determine all of the amplitudes $\{\alpha_j\}$.

Likewise, for model \eqref{eq:real_const_model}, we have
\begin{align}
\left[\begin{smallmatrix}
F(1,1) & F(1,2) & \cdots & F(1,T) \\
\vdots & \vdots & & \vdots \\
F(n,1) & F(n,2) & \cdots & F(n,T)
\end{smallmatrix}\right]
=\left[\begin{smallmatrix}
\alpha_0(1) & \alpha_1(1) & \bar{\alpha}_1(1) & \cdots & \bar{\alpha}_N(1) \\
\vdots & \vdots & \vdots & & \vdots \\
\alpha_0(n) & \alpha_1(n) & \bar{\alpha}_1(n) & \cdots & \bar{\alpha}_N(n)
\end{smallmatrix}\right]
\left[\begin{smallmatrix}
1 & 1 & \cdots & 1 \\
z_1 & z_1^2 & \cdots & z_1^T \\
\bar{z}_1 & \bar{z}_1^2 & \cdots & \bar{z}_1^T \\
\vdots & \vdots & & \vdots \\
\bar{z}_N & \bar{z}_N^2 & \cdots & \bar{z}_N^T
\end{smallmatrix}\right],
\label{eq:real_const_amp}
\end{align}
for any $T\in\Z_{>0}$.
Hence, the observed values at $2N+1$ integer points $\{F(\cdot,t)\in\CA_V\mid t = 1, 2, \cdots, 2N+1\}$ determine all of the amplitudes $\{\alpha_j\}$.
More generally, even for a dataset not generated from model \eqref{eq:real_const_model}, we can obtain this expansion.
In other words, the following lemma holds.
\begin{lem}
\label{lem:complex_amp}
Any $2N+1$ real values $\{ F_t \in \R \mid t = 1, 2, \cdots, 2N+1 \}$ are represented as
\begin{align*}
F_t = a_0z_0 + \sum_{j=1}^{N}\{a_jz_j^t + a_{j+N}\bar{z}_j^t\}, \text{ for }t = 1, 2, \cdots, 2N+1,
\end{align*}
by distinct modes $z_0=1$ and $\{z_j, \bar{z}_j\mid 1\leq j \leq N\}$, and $a_j \in \C$.
Moreover, $a_0\in \R$ and $a_{j+N} = \bar{a}_j$ hold.
\end{lem}
\begin{proof}
Take $(2N+1)$-dimensional vectors
\begin{align*}
\mb{v}_0 = (1, 1, \cdots, 1), \ \mb{v}_j = (z_j, z_j^2, \cdots, z_j^{2N+1}), \ \mb{v}_{j+N} = (\bar{z}_j, \bar{z}_j^2, \cdots, \bar{z}_j^{2N+1})
\end{align*}
for $1 \leq j \leq N$.
These vectors consist of a basis of $\C^{2N+1}$ because of the distinct condition.
Then, we have their dual vectors $\{ \mb{v}_j^* \mid 0 \leq j \leq 2N\}$ such that $\langle \mb{v}_j, \mb{v}^*_l \rangle = \delta_{j,l}$, which are given by the inverse matrix of the matrix formed by $\{ \mb{v}_j \}$.
Therefore, $a_j$ is defined by $\langle \mb{F}, \mb{v}^*_j \rangle$ for $\mb{F} = (F_1,\cdots,F_{2N+1})$.
Moreover, for $1 \leq j \leq N$, we have $\bar{a}_j = \langle \mb{F}, \bar{\mb{v}}^*_j \rangle = \langle \mb{F}, \mb{v}^*_{j+N} \rangle = a_{j+N}$ because $\bar{\mb{v}}^*_j$ is the dual vector of $\bar{\mb{v}}_j = \mb{v}_{j+N}$.
\end{proof}

In either case, the number of observation points needed to determine the amplitudes is less than that needed to determine the modes shown in Propositions \ref{prop:prony} and \ref{prop:real_const_prony}.
Therefore, the mode expansion techniques are summarized as follows:
\begin{thm}
\label{thm:mode_extraction}
\emph{(i)} General model \eqref{eq:general_model} is determined from $n(L+M)$ complex values $\{F(\x,t)\in \C\mid \x \in V, t = 1, 2, \cdots, L+M\}$ for $L=\max(M-n,0)+1$ as in Proposition \ref{prop:prony}.
When $n=1$ or $M$, this number coincides with the complex dimension $M(n+1)$ of the model parameter.

\emph{(ii)} Particular model \eqref{eq:real_const_model} is determined from $n(L+2N+1)$ real values $\{F(\x,t)\in \R\mid \x \in V, t = 1, 2, \cdots, L+2N+1\}$ for $L=\max(N-n,0)+1$, as in Proposition \ref{prop:real_const_prony}.
When $n=1$ or $N$, this number coincides with the real dimension $(2N+1)n+N$ of the model parameter.
\end{thm}
We can use the estimated mode expansion for forecasting unobserved data.
\begin{cor}
\label{cor:forecasting}
When $\{\alpha_j, z_j\}$ are determined from given observed points in Theorem \ref{thm:mode_extraction}, further values $\{ F(\x,t) \}$ are computed for any $t\in \Z$.
\end{cor}

\subsection{Example}
\label{subsec:example}
Here, let us consider the following $n=1$ time-series data for example;
\begin{figure}[t]
\begin{center}
\includegraphics[height=4cm]{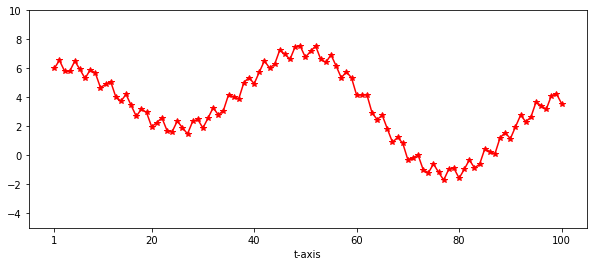}
\caption{Example data in \eqref{eq:raw_data}.}
\label{fig:raw_data}
\end{center}
\end{figure}
\begin{align}
F(t) = 2\sin(0.016\pi t) + 3\cos(0.04\pi t) + \frac{1}{2}\sin(0.6\pi t) + 3,
\label{eq:raw_data}
\end{align}
for $t = 1, 2, \cdots, 100$.
See the plots in Figure \ref{fig:raw_data}.

\begin{figure}[t]
\begin{center}
\includegraphics[height=4cm]{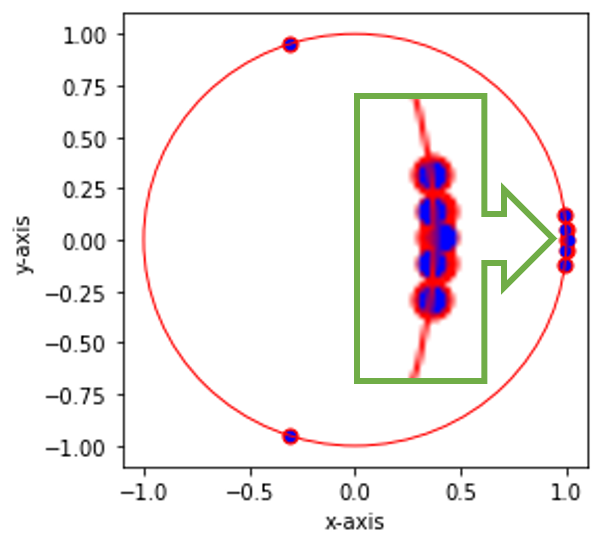}
\ 
\includegraphics[height=4cm]{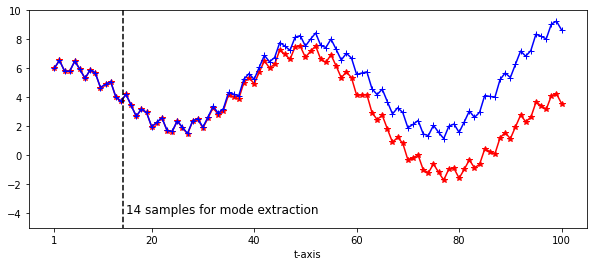}
\caption{[left] True mode set (red) and estimated mode set (blue) by Proposition \ref{prop:prony}.
[right] The blue plot indicates a forecast after $t=14$ by using the estimated modes and amplitudes overlaid on the red true plot in Figure \ref{fig:raw_data}.}
\label{fig:DMD_modes}
\end{center}
\end{figure}
By applying Proposition \ref{prop:prony} for $M = 7$ and $L = 7$ to $\{F(t)\in\R\mid t = 1, 2, \cdots, 14\}$, we estimate modes and amplitudes, and then forecast the remaining values for $t > 14$.
As is well known, Prony's method is too sensitive to extract precise modes from numerical values, given as 64-bit floating point numbers in our case.
Therefore, the results in Figure \ref{fig:DMD_modes} show that the estimated modes are slightly different from the true values, and then such error accumulates to give a bad forecast.

\begin{figure}[t]
\begin{center}
\includegraphics[height=4cm]{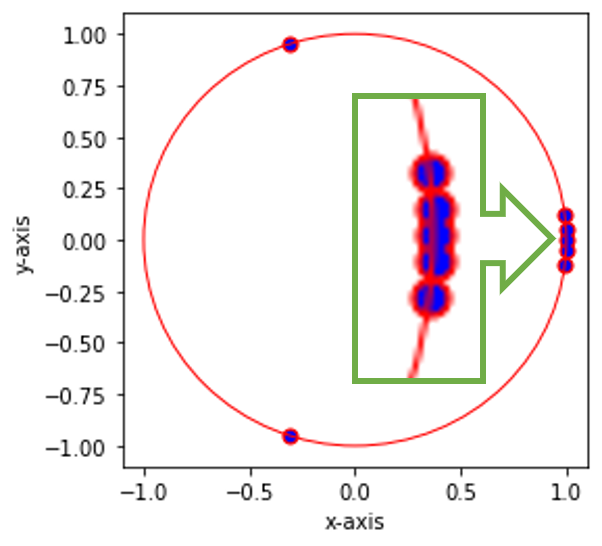}
\ 
\includegraphics[height=4cm]{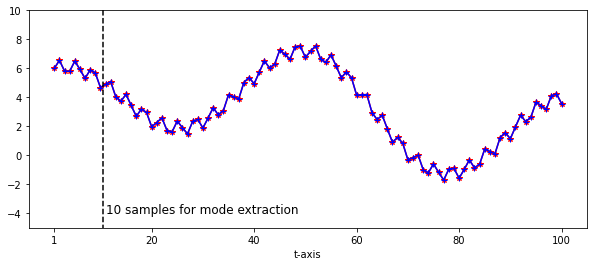}
\caption{[left] True mode set (red) and estimated mode set (blue) by Proposition \ref{prop:real_const_prony}.
[right] The blue forecasting plot after $t=10$ lies exactly on the red true plot in Figure \ref{fig:raw_data}.}
\label{fig:DMD-DFT_modes}
\end{center}
\end{figure}
On the other hand, Proposition \ref{prop:real_const_prony} for $N = 3$ and $L = 3$ gives better results from a smaller number of values $\{F(t)\in\R\mid t = 1, 2, \cdots, 10\}$, as in Figure \ref{fig:DMD-DFT_modes}.
This is because the absolute values of extracted modes are controlled to be $1$, as explained in Figure \ref{fig:modes}.

In \S \ref{sec:experiment}, we explain a more detailed algorithm with a graph recovery procedure.

\section{Graph wave equation and graph recovery}
\label{sec:graphwave}
In this section, we start from a review of the graph Laplacian and the graph wave equation from \S \ref{subsec:laplacian} to \S \ref{subsec:graph_wave}.
In order to see an explicit solution of the graph wave equation, we calculate eigenfunctions for a path graph in \S \ref{subsec:fourier}, which is used in \S \ref{sec:experiment}.
Finally, we solve a graph recovery problem from the graph wave functions in \S \ref{subsec:graph_recovery}.

\subsection{Graph Laplacian}
\label{subsec:laplacian}
Let $\{\ww_{\x\y}\}$ be graph weights that satisfy $\ww_{\y\x} = \ww_{\x\y} \geq 0$ and $\ww_{\x\x}=0$ for $\x\neq\y \in V$.
In the present paper, we regard $(V, \{\ww_{\x\y}\})$ as a weighted undirected graph that has an edge $\x$-$\y$ if and only if $\ww_{\x\y} > 0$.

For a function $f\in \CA_V$, we define \emph{the graph Laplacian} $\lap\colon \CA_V\rightarrow \CA_V$ as
\begin{align}
\lap f(\x) := \sum_{\y \in V}\ww_{\x\y}(f(\x) - f(\y))  = \deg(\x)f(\x) - \sum_{\y\in V}\ww_{\x\y}f(\y), \label{eq:laplacian}
\end{align}
where $\deg(\x) := \sum_{\y\in V}\ww_{\x\y}$.
Note that graph weights $\{\ww_{\x\y}\}$ are defined as a metric of $1$-forms over a discrete set $V$, and the graph Laplacian is given as the Laplace-Beltrami operator with an inner product $\langle f, g\rangle_V:=\sum_{\x\in V}f(\x)g(\x)$ \cite{takayama2020geometric}.
Since the Laplacian $\lap$ is self-adjoint for this inner product, we can take eigenfunctions $\{v_k\in \CA_V \mid 0 \leq k \leq n-1\}$ as follows:
\begin{align}
\lap v_k = \rho_k v_k, \ 0 = \rho_0 \leq \rho_1 \leq \cdots \leq \rho_{n-1}, \ \langle v_k,v_l \rangle_V = \delta_{k,l}.
\label{eq:eigenfunctions}
\end{align}
Here, we see $v_0 = n^{-1/2}1_V$ and $v_k \perp 1_V$ for $k \geq 1$.
Although the eigenfunctions $\{v_k\}$ have ambiguity regarding their signs even when the eigenvalues $\{ \rho_k \}$ are distinct, we assume that the eigenfunctions are fixed in a certain manner.
Moreover, in the present paper, we assume that a weighted graph $(V, \{\ww_{\x\y}\})$ is always connected, so that $\rho_1 > 0$.
The upper and lower bounds of the eigenvalues are well studied; for example, we have
\begin{align}
\rho_{n-1} \leq 2\max_{\x\in V}\deg(\x).
\label{eq:upper_bound}
\end{align}
The proof and other bounds of graphs are shown in \cite{chung1997spectral}.

\subsection{Graph Fourier transform}
\label{subsec:fourier}
Set $\CF[f]_k := \langle f, v_k \rangle_V \in \R$ for $f \in \CA_V$.
Then, the eigenfunction expansion is given as
\begin{align}
f = \sum_{k=0}^{n-1} \langle f, v_k \rangle_V v_k = \sum_{k=0}^{n-1} \CF[f]_k v_k.
\label{eq:fourier}
\end{align}
When the eigenvalues $\{ \rho_k \}$ are distinct, this expansion is unique.
In signal processing, the transformation $\CF\colon \CA_V \ni f \mapsto \CF[f] \in \R^{n}$ is called the \emph{graph Fourier transform} \cite{hammond2011wavelets}.
For the support function $\chi_{\x}$ on $\x\in V$, i.e., $\chi_{\x}(\y) = \delta_{\x,\y}$, we have $\CF[\chi_{\x}]_k = \langle \chi_{\x}, v_k \rangle_V = v_k(\x)$.
This calculation formally rewrites the eigenfunction decomposition of the graph Laplacian $\lap$ in \eqref{eq:eigenfunctions} as the following identity concerning the graph weight:
\begin{align}
-\ww_{\x\y} &= \langle \chi_{\x},\lap\chi_{\y}\rangle_{V} \notag \\
&= \sum_{k=0}^{n-1}\langle \CF[\chi_{\x}]_kv_k, \rho_k\CF[\chi_{\y}]_kv_k\rangle_{V} = \sum_{k=0}^{n-1}v_k(\x)\rho_kv_k(\y).
\label{eq:weight_identity}
\end{align}

\begin{exa}[Path graph]
\label{ex:path_graph}
Let us consider the \emph{path graph} $P_n$ in Figure \ref{fig:path_graph}, which is defined by the weights
\begin{align}
w_{\x\y} = \begin{cases}
1 & \text{for $\x, \y \in V$ such that $\y = \x \pm 1$}\\
0 & \text{otherwise.}
\end{cases}
\label{eq:path_weight}
\end{align}
\begin{figure}[t]
\begin{center}
\includegraphics[width=4cm]{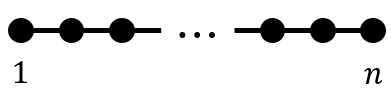}
\caption{Path graph $P_n$ in Example \ref{ex:path_graph}.}
\label{fig:path_graph}
\end{center}
\end{figure}
The eigenvalues and eigenfunctions are explicitly calculated as
\begin{align}
\rho_k = 2-2\cos \frac{\pi k}{n}, \ v_k(\x) = \sqrt{\frac{2}{n}}\cos \frac{\pi k(2\x-1)}{2n},
\label{eq:path_eigen}
\end{align}
for $1 \leq k \leq n-1$.
This is directly checked via the identity
\begin{align}
\exp \frac{\pi\i k(2\x+1)}{2n} + \exp \frac{\pi\i k(2\x-3)}{2n} = 2\cos \frac{\pi k}{n}\exp \frac{\pi\i k(2\x-1)}{2n}
\label{eq:explicit_eigen}
\end{align}
and so on.
Here, as the upper bound \eqref{eq:upper_bound}, the largest eigenvalue satisfies $\rho_{n-1} < 4$.
See also Figure \ref{fig:eigen} for $n=21$.
The corresponding graph Fourier transform is known as the \emph{discrete cosine transform} in image processing \cite{asmara2017comparison}.
\end{exa}

Similarly, the graph Fourier transform for a cycle graph is known as the \emph{discrete Fourier transform} (DFT) in signal processing and image processing \cite{asmara2017comparison}.

\subsection{Graph wave equation}
\label{subsec:graph_wave}
By use of the graph Laplacian, we can define the following \emph{graph wave equation} \cite{friedman2004wave} for $U \in \CA_V\otimes\CC^2(\R; \R)$:
\begin{align}
\frac{d^2}{dt^2} U(\x,t) = -c\lap U(\x,t), \ U(\x,0) = f(\x), \ \frac{d}{dt}\Big|_{t=0} U(\x,t) = g(\x),
\label{eq:graphwave}
\end{align}
for a certain $c > 0$ and given $f, g \in \CA_V$.
For a solution $U(\x,t)$, the graph Fourier coefficient $\CF[U(\cdot,t)]$ becomes the harmonic oscillator
\begin{align*}
\frac{d^2}{dt^2}\CF[U(\cdot,t)]_k = -c\rho_k\CF[U(\cdot,t)]_k, \ \CF[U(\cdot,0)]_k = \CF[f]_k, \ \frac{d}{dt}\Big|_{t=0} \CF[U(\cdot,t)]_k = \CF[g]_k,
\end{align*}
for $0 \leq k \leq n-1$; hence, we can describe $U(\x,t)$ as
\begin{align}
U(\x,t) = \CF[f]_0 v_0(\x) + \sum_{k=1}^{n-1} \left\{ \CF[f]_k \cos{\sqrt{c\rho_k} t} + \frac{\CF[g]_k}{\sqrt{c\rho_k}} \sin{\sqrt{c\rho_k} t} \right\}v_k(\x).
\label{eq:wave_function}
\end{align}
Therefore, graph wave equation \eqref{eq:graphwave} has a unique solution if and only if $g \perp 1_V$.
For simplicity, we call this $U(\x, t)$ a \emph{graph wave function}.
In the present paper, we also assume that the Nyquist condition is satisfied:
\begin{align}
\sqrt{c\rho_k} < \pi, \text{ for any }k.
\label{eq:Nyquist}
\end{align}
This is always achieved by re-parameterizing $t$ by a higher sampling rate than $1$, which also makes $c$ smaller, if necessary.

Here, it is worth noting that the above graph wave equation is naturally regarded as a discretization of the (continuous) wave equation
\begin{align*}
\frac{\d^2}{\d t^2} W(x,t) = c\sum_{i=1}^d\frac{\d^2}{\d x_i^2}W(x,t)
\end{align*}
for $W \in \CC^2(\R^d\times\R; \R)$, which is easily checked for the path graph $P_n$.
Namely, for the weights \eqref{eq:path_weight}, the graph Laplacian gives the second-order central difference in the spatial domain
\begin{align*}
-\lap f(\x) = f(\x+1) + f(\x-1) - 2f(\x) \sim \frac{\d^2}{\d x^2}f(x).
\end{align*}
In order to see boundary conditions, we compare continuous and graph wave functions directly.
The wave equation over $[0, n]$ with the Neumann conditions is written as
\begin{align*}
\frac{\d^2}{\d t^2} W(x,t) = c\frac{\d^2}{\d x^2} W(x,t), \ \frac{\d}{\d x}\Big|_{x=0} W(x,t) = 0, \ \frac{\d}{\d x}\Big|_{x=n} W(x,t) = 0,
\end{align*}
and its solution is explicitly described with some $\{a_k, b_k\}$ by
\begin{align}
W(x,t) = a_0 + \sum_{k=1}^{\infty} \left\{ a_k \cos\frac{\sqrt{c}\pi kt}{n} + b_k \sin\frac{\sqrt{c} \pi kt}{n} \right\}\cos\frac{\pi kx}{n}.
\label{eq:cont_wave_function}
\end{align}
\begin{figure}[t]
\begin{center}
\includegraphics[width=7cm]{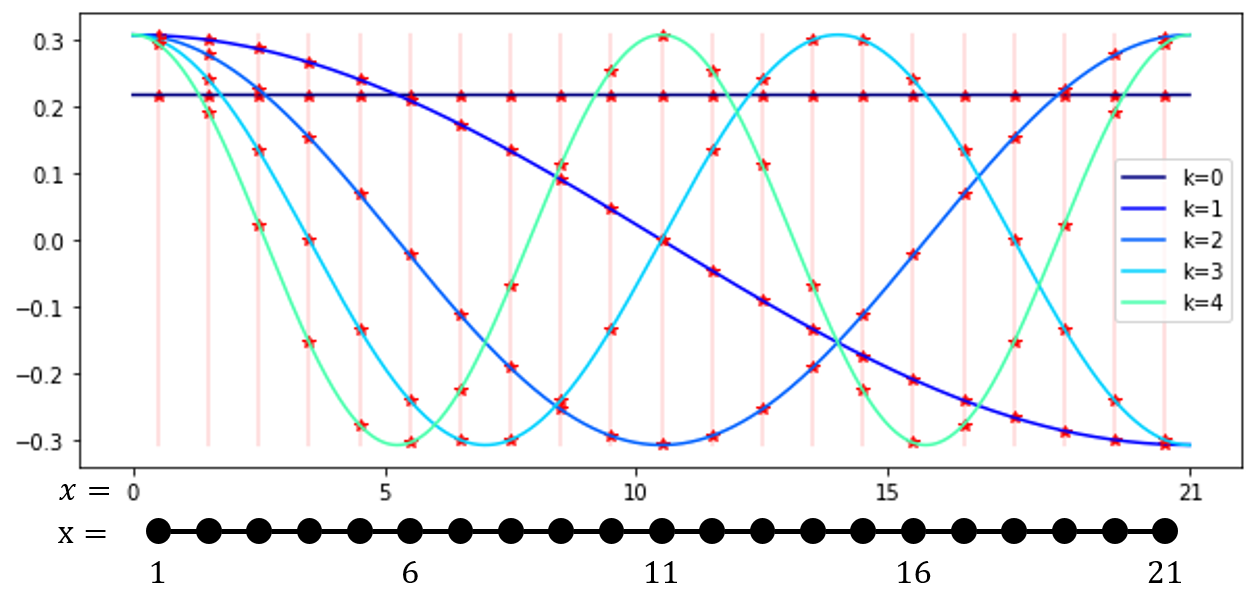}
\
\includegraphics[width=7cm]{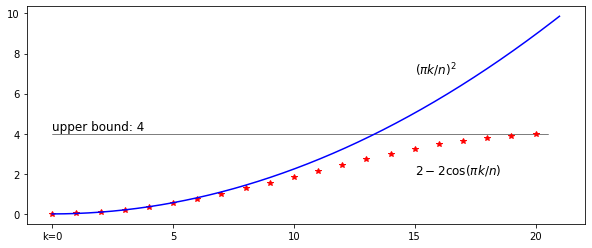}
\caption{For $n=21$, eigenfunctions and eigenvalues are plotted.
[left] Correspondence among five graph eigenfunctions (red dots) and the five continuous eigenfunctions (blue curves).
[right] Correspondence among the graph eigenvalues (red dots) and the square of the continuous frequencies (blue curves).}
\label{fig:eigen}
\end{center}
\end{figure}
This is viewed as a continuous analogue of the solution \eqref{eq:wave_function} for the path graph $P_n$, because the eigenfunctions $\{\cos \frac{\pi k(2\x-1)}{2n}\}$ in \eqref{eq:path_eigen} are obtained from $\{\cos\frac{\pi kx}{n}\}$ by discretization, and the eigenvalues $\{\rho_k\}$ in \eqref{eq:path_eigen} are computed by
\begin{align*}
\sqrt{\rho_k} = \sqrt{2-2\cos\frac{\pi k}{n}} = 2\sin \frac{\pi k}{2n} \sim \frac{\pi k}{n}
\end{align*}
for $k \ll n$.
These correspondences are illustrated in Figure \ref{fig:eigen}.
\begin{rem}
Even when a continuous wave function $W(x,t)$ in \eqref{eq:cont_wave_function} is evaluated at discrete points, like $x = 1/2n, 3/2n, \cdots, (2n-1)/2n$, the points generally do not satisfy the graph wave equation \eqref{eq:graphwave} over a path graph.
We see this point at \S \ref{subsec:cont_wave_demo}.
\end{rem}

\subsection{Graph recovery from graph wave function}
\label{subsec:graph_recovery}
The target in the present paper is to give a procedure to recover the graph weights $\{\ww_{\x\y}\}$ from finite observed values at integer points $\{ U(\cdot, t) \in \CA_V \mid t = 1, 2, \cdots, T \}$.
To state our main theorem, we simplify the representation of a graph wave function $U$ in \eqref{eq:wave_function} as
\begin{align}
U(\x,t) = \beta_0 v_0(\x) + \sum_{k=1}^{n-1} \left\{ \beta_k z_k^t + \bar{\beta}_k \bar{z}_k^{t}\right\}v_k(\x)
\label{eq:simple_wave}
\end{align}
where $\beta_0 \in \R$, $\beta_k \in \C$, and $z_k=e^{\i\sqrt{c\rho_k}}$.
\begin{thm}
\label{thm:weight_recover}
\emph{(i)} Assume that the modes $\{z_k\}$ in \eqref{eq:simple_wave} are distinct and unknown.
If $\beta_k\neq 0$ for all $k$, then the modes $\{z_k\}$, the eigenfunctions $\{v_k\}$ and the amplitudes $\{\beta_k\}$ are determined from observed valued at $2n$ integer points $\{ U(\cdot,t) \in \CA_V \mid t = 1, 2, \cdots, 2n \}$.

\emph{(ii)} The graph Laplacian $\lap$ and the graph weights $\{\ww_{\x\y}\}$ are recovered from the given modes $\{z_k\}$ and the given eigenfunctions $\{v_k\}$ up to constant $c$.
\end{thm}
\begin{proof}
The modes and amplitudes are determined by Theorem \ref{thm:mode_extraction}, because \eqref{eq:simple_wave} is a special case of the time-series model \eqref{eq:real_const_model} for $N=n-1$ and $\alpha_k = \beta_kv_k$, where non-degenerate condition \eqref{eq:non-degenerate2} is satisfied.
Then, from amplitudes $\{\beta_kv_k\}$, the eigenfunctions $\{v_k\}$ are determined by $\|v_k\|=1$, as asserted in (i).

For (ii), note that Nyquist condition \eqref{eq:Nyquist} distinguishes $z_k$ and $\bar{z}_k$ by the signs of their imaginary parts.
Then, $\sqrt{c\rho_k} = \arg z_k$ is defined uniquely.
Hence, the assertion follows from the identity \eqref{eq:weight_identity}.
\end{proof}
In theory, the underlying graph is determined by Theorem \ref{thm:weight_recover}.
On the other hand, in numerical computation, it is often difficult to compute the underlying graph exactly because of numerical error, although we do not assume any noise.
Therefore, in the following section, we summarize all of the above procedures as an algorithm.

\section{Numerical experiments}
\label{sec:experiment}
In this section, we summarize the previous arguments into an algorithm, and then determine its performance using numerical datasets.
In \S \ref{subsec:algorithm}, we provide a six-step algorithm originating from Theorems \ref{thm:mode_extraction} and \ref{thm:weight_recover}.
Then, in \S \ref{subsec:graph_wave_demo}, \S \ref{subsec:cont_wave_demo}, and \S \ref{subsec:pose_model}, we apply the algorithm to graph wave signals, continuous wave signals, and human joint sensor time-series data, respectively.

\subsection{Graph recover algorithm}
\label{subsec:algorithm}
For given multivariate wave signals, the number $n$ of variables is obviously known.
Then, $N = n-1$, $L = 1$, and $T = 2n$ in Proposition \ref{prop:real_const_prony} and \eqref{eq:real_const_amp} should be optimal parameters to recover the underlying graph as the proof of Theorem \ref{thm:weight_recover}.
In addition, each calculated amplitude $\alpha_k\in \CA_V\otimes\C$ is supposed to be decomposed into $\beta_k \in \C$ and $v_k \in \CA_V$ as $\alpha_k = \beta_kv_k$.
However, in practice, numerical error causes unexpected results in spite of Theorem \ref{thm:weight_recover}.
Therefore, we propose the following algorithm.
\begin{algo}
\label{algo:algorithm}
Execute the following procedure:
\begin{description}
\item[Step 1] Set $N, L \in \Z_{> 0}$ such that $nL \geq N$, and take $n(L+2N+1)$ real values $\{F_{\x, t} \in \R \mid 1 \leq \x \leq n, t = 1, 2, \cdots, L+2N+1\}$.
\item[Step 2] Define $G(\x,t;j)$ as in Lemma \ref{lem:real_const_dmd}, and get $\{\hat{d}_j\in\R \mid 1\leq j \leq N\}$ by solving linear equation \eqref{eq:real_const_dmd_matrix}.
\item[Step 3] Obtain $\{\hat{z}_j \in\C \mid 1 \leq j \leq 2N\}$ by solving the polynomial equation
\begin{align}
Z^{2N} - \hat{d}_{1}Z^{2N-1} - \cdots - \hat{d}_{N-1}Z^{N+1} - \hat{d}_NZ^{N} - \hat{d}_{N-1}Z^{N-1} \cdots - \hat{d}_{1}Z +1 =0
\label{eq:real_polynomial}
\end{align}
with the above $\{\hat{d}_j\}$, and compute $\hat{\theta}_j := \arg \hat{z}_j$ so that $-\pi < \hat{\theta}_j < \pi$ for $1 \leq j \leq 2N$.
Moreover, set $\hat{z}_0 := 1$ and $\hat{\theta}_0 := 0$.
\item[Step 4] Get $\{\hat{\alpha}_j\in\CA_V\otimes\C\mid 0 \leq j \leq 2N\}$ by solving linear equation \eqref{eq:amp} for $T=L+2N+1$ with the above $\{\hat{z}_j\}$.
\item[Step 5] If a reconstructed $\hat{F}(\x,t) :=\sum_{j=0}^{2N}\hat{\alpha}_j(\x)\hat{z}_j^t$ is far from a validation dataset, then go back to Step 1 and retake larger $N$ and $L$. 
\item[Step 6] If $\|\hat{\alpha}_j\|$ is tiny, then set $\hat{v}_j := 0_V$, otherwise $\hat{v}_j := \hat{\alpha}_j/\|\hat{\alpha}_j\|$.
Then, define
\begin{align}
\hat{\ww}_{\x\y} := -\frac{1}{2}\sum_{j=1}^{2N}\Re\left(\hat{v}_j(\x)\hat{\theta}_j^2\bar{\hat{v}}_j(\y)\right)
\label{eq:weight_estimate}
\end{align}
for any $\x\neq\y \in V$ from the above $\{\hat{\theta}_j\}$.
Optionally, we can retake this weight as $\max(\hat{\ww}_{\x\y},0)$ instead if the graph weights are assumed to be non-negative.
\end{description}
\end{algo}
Here, we present some remarks on this algorithm.
For the graph recovery in Theorem \ref{thm:weight_recover}, $N$ is supposed to be larger than $n-1$ as in examples \S \ref{subsec:graph_wave_demo} and \S \ref{subsec:cont_wave_demo}.
On the other hand, for a practical application, such as \S \ref{subsec:pose_model}, smaller $N$ may work better to extract a significant graph.
The condition $nL \geq N$ means linear equation  \eqref{eq:real_const_dmd_matrix} is an overdetermined system.
Hence, $\{\hat{d}_j\in\R\mid 1\leq j \leq N\}$ are computed by the least squares method or the pseudo-inverse matrix of the rectangular Vandermonde matrix.
Then, modes $\{\hat{z}_j\}$ are calculated as solutions of polynomial equation \eqref{eq:real_polynomial}.
For example, this part is practically computed as eigenvalues of the companion matrix of coefficients $\{\hat{d}_j\}$.
As an expression of graph wave function \eqref{eq:wave_function}, $\hat{\theta}_j^2 = (\arg\hat{z}_j)^2$ corresponds to the eigenvalue $\hat{\rho}_j$.
If constant $c$ in graph wave equation \eqref{eq:graphwave} is known, then $\hat{\theta}$ can be rescaled.
Since the absolute value of $\hat{z}_j$ is not always $1$ as explained in Remark \ref{rem:stationary_mode}, $\hat{z}_j$ is sometimes retaken as $\hat{z}'_j := \hat{z}_j/|\hat{z}_j|$.
In either case, for $\hat{z}_j$, we have $\hat{z}_l = \bar{\hat{z}}_j$ for a certain $l$.
Then linear equation \eqref{eq:amp} becomes \eqref{eq:real_const_amp}.
Hence, by Lemma \ref{lem:complex_amp}, given values $\{F_{\x,t}\}$ reconstruct a function $\hat{F}$ as
\begin{align}
\hat{F}(\x,t) := \sum_{j=0}^{2N}\hat{\alpha}_j(\x)\hat{z}_j^t = \hat{\alpha}_0(\x) + \frac{1}{2}\sum_{j=1}^{2N}\{\hat{\alpha}_j(\x)z_j^t + \bar{\hat{\alpha}}_j(\x)\bar{z}_j^t\}
\label{eq:reconstruct}
\end{align}
for any $t \in \Z$.
In an ideal case, $\hat{\alpha}_j$ is decomposed as $\hat{\beta}_j \hat{v}_j$ by $\hat{\beta}_j \in \C$, and then the weights are estimated from identity \eqref{eq:weight_identity}.
Instead, we set $\hat{\alpha}_j =: \|\hat{\alpha}_j\|\hat{v}_j$ so that $\| \hat{v}_j \|=1$.
Since $\hat{v}_j$ is in $\CA_V\otimes\C$ in general, we replace $v_j(\x)\rho_j{v}_j(\y)$ by $\Re \hat{v}_j(\x)\hat{\rho}_j\bar{\hat{v}}_j(\y)$ as in \eqref{eq:weight_estimate}, in order to eliminate the ambiguity of an unit complex value on $\hat{v}_j$.
Basically, this $\Re \hat{v}_j(\x)\hat{\rho}_j\bar{\hat{v}}_j(\y)$ depends only on $\hat{\alpha}_j(\x)$ and $\hat{\alpha}_j(\y)$ up to the constant $\|\hat{\alpha}_j\|$; thus, partial observed values along $V$ are supposed to give a subgraph.
We check this in \S \ref{subsec:graph_wave_demo}.

Furthermore, we consider the meaning of the weight \eqref{eq:weight_estimate}.
\begin{figure}[t]
\begin{center}
\includegraphics[height=4cm]{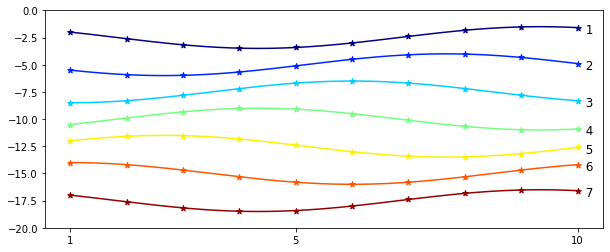}
\ 
\includegraphics[height=4cm]{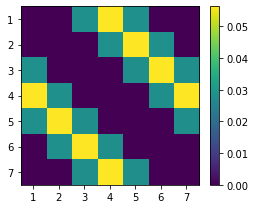}
\caption{[left] Sample multivariate function given in \eqref{eq:weight_sample}.
[right] Estimated weights for the sample.}
\label{fig:weight_property}
\end{center}
\end{figure}
Since $\hat{\theta}_0=0$, the constant term $\hat{\alpha}_0$ in \eqref{eq:reconstruct} does not affect the weights $\{\hat{\ww}_{\x\y}\}$.
The other term $\hat{\alpha}_j(\x)$ indicates the existence of the mode $z_j$ at node $\x$.
Hence, when $F_{\x}$ and $F_{\y}$ have no common mode, they have no connection, i.e., $\hat{\ww}_{\x\y}=0$.
Moreover, when they are synchronized, $\hat{\ww}_{\x\y}$ becomes negative, and when they are anti-synchronized, $\hat{\ww}_{\x\y}$ becomes positive.
To examine this property, let us consider the following example:
\begin{align}
F(\x,t) &= -\frac{5}{2}\x + \cos\left(\frac{\pi}{3}\x + \frac{\pi}{5}t \right) \notag \\
&= -\frac{5}{2}\x + \frac{1}{2}\exp\left(\i\frac{\pi}{3}\x\right)\exp\left(\i\frac{\pi}{5}t \right) + \frac{1}{2}\exp\left(-\i\frac{\pi}{3}\x\right)\exp\left(-\i\frac{\pi}{5}t \right) 
\label{eq:weight_sample}
\end{align}
for $1 \leq \x \leq 7$.
We have $\hat{\theta}_1 = \pi/5$ and $\hat{v}_1(\x) = \exp(\i\pi\x/3)/\sqrt{7}$, and then positivized weights $\{\max(\hat{\ww}_{\x\y},0)\}$ are calculated as in Figure \ref{fig:weight_property}.
There, the maximum value of the weights is $(\pi/5)^2/7 \fallingdotseq 0.0564$.

\subsection{Graph wave over an interval}
\label{subsec:graph_wave_demo}
We hereinafter demonstrate Algorithm \ref{algo:algorithm} in order to verify its effectiveness.
The first example is a graph wave function on an interval, as explained in \S \ref{subsec:graph_wave}.
For $n=21$ and $V = \{ 1 \leq \x \leq 21\}$, let us take an initial function $f$ in the graph wave equation \eqref{eq:graphwave} over the path graph $P_{21}$ as
\begin{align}
\tilde{f}(\x) = \frac{\x-11}{5} + \text{(small random value)}, \text{ for }\x \in V
\label{eq:wave_init}
\end{align}
and $\tilde{g} = 0_V$.
\begin{figure}[t]
\begin{center}
\includegraphics[width=7cm]{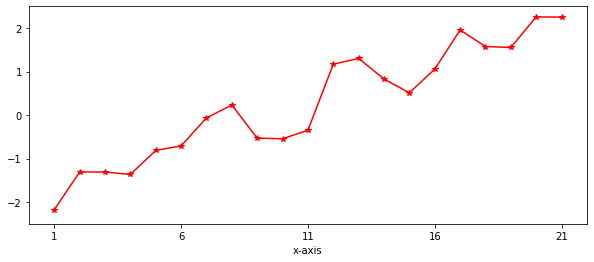}
\ 
\includegraphics[width=7cm]{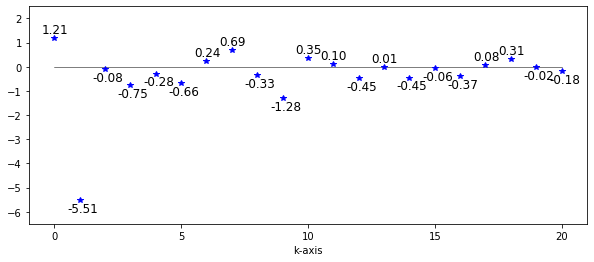}
\caption{[left] Initial function given in \eqref{eq:wave_init} for the graph wave equation.
[right] Corresponding graph Fourier transform.}
\label{fig:graph_wave_init}
\end{center}
\end{figure}

Note that $\{\CF[\tilde{f}]_k\}$ satisfy the assumption in Theorem \ref{thm:weight_recover}, as in Figure \ref{fig:graph_wave_init}.
Moreover, we take $\sqrt{c}=1.5$ so that even the largest eigenvalue satisfies Nyquist condition \eqref{eq:Nyquist}:
\begin{align*}
\sqrt{c\rho_k} < 1.5\times 2 < \pi,
\end{align*}
where upper bound \eqref{eq:upper_bound} is used.
Then, by applying the expression \eqref{eq:wave_function} to these conditions, we can see that the wave behavior becomes like the water in a bottle as in Figure \ref{fig:graph_wave}.
Although the function is continuous with respect to time variable $t$ as on the left-hand side, only integer points are used in the algorithm shown as the heat map on the right-hand side.
\begin{figure}[t]
\begin{tabular}{cc}
\begin{minipage}[c]{0.45\hsize}
\begin{center}
\includegraphics[height=5cm]{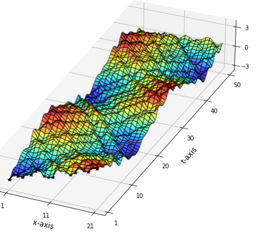}
\end{center}
\end{minipage}
&
\begin{minipage}[c]{0.45\hsize}
\begin{center} 
\includegraphics[height=3cm]{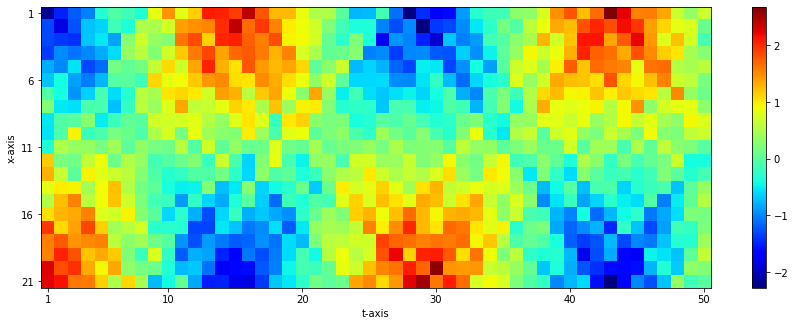}
\end{center}
\end{minipage}
\end{tabular}
\caption{[left] Graph wave function given from initial conditions \eqref{eq:wave_init} and $\tilde{g}=0_V$ over the path graph $P_{21}$.
[right] Heatmap representation by discretization.}
\label{fig:graph_wave}
\end{figure}

To this dataset, we apply Algorithm \ref{algo:algorithm}.
\begin{figure}[tp]
\begin{center}
\includegraphics[height=4.5cm]{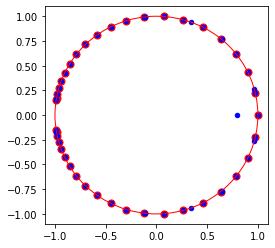}
\ 
\includegraphics[height=4.5cm]{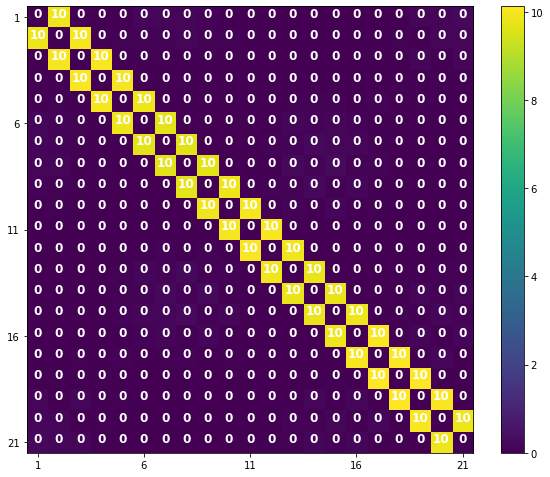}
\caption{[left] Forty-one true modes (red) and 53 estimated modes (blue).
[right] Recovered graph structure (10x), which is supposed to be the path graph $P_{21}$ in \eqref{eq:path_weight}.}
\label{fig:graph_recover}
\end{center}
\end{figure}
First, let us take $L = 3$ and $N = 26$ in Step 1.
Then, we obtain $53$ modes $\{\hat{z}_j\}$ from $56$ observation points along the time axis from Steps 2 and 3, as on the left-hand side in Figure \ref{fig:graph_recover}.
Note that although $53$ is larger than $41$ (the number of actual modes), such redundancy gives better results in numerical computation.
Then, we obtain the corresponding amplitudes $\{\hat{\alpha}_j\}$ and weights $\{\hat{\ww}_{\x\y}\}$ from Steps 4 and 6.
Since $c$ is given as above, $\hat{\theta}_j = \arg \hat{z}_j/\sqrt{c}$ can be used to calculate weights instead of $\hat{\theta}_j$ in Step 6.
On the right-hand side in Figure \ref{fig:graph_recover}, we multiply the recovered weights by $10$ for visualization, which is quite close to the original weights in the path graph \eqref{eq:path_weight}.
Finally, we check the mean squared error in each $t$ as in Figure \ref{fig:graph_recover_mse} to verify Step 5.
\begin{figure}[tp]
\begin{center}
\includegraphics[width=7cm]{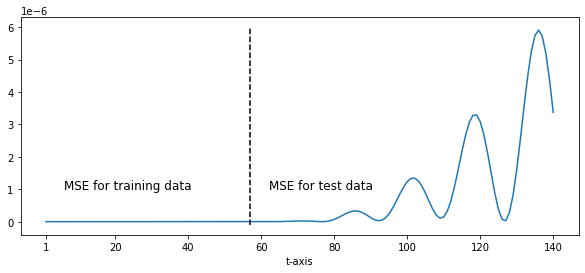}
\caption{Mean squared error between exact data and reconstructed data.}
\label{fig:graph_recover_mse}
\end{center}
\end{figure}
Since the error is tiny, even for test data, we conclude that the proposed graph recovery algorithm works well for this exact graph wave setting.

By modifying this dataset, let us consider another case when only partial observed values with respect to spatial variables are available.
\begin{figure}[tp]
\begin{center}
\includegraphics[height=3cm]{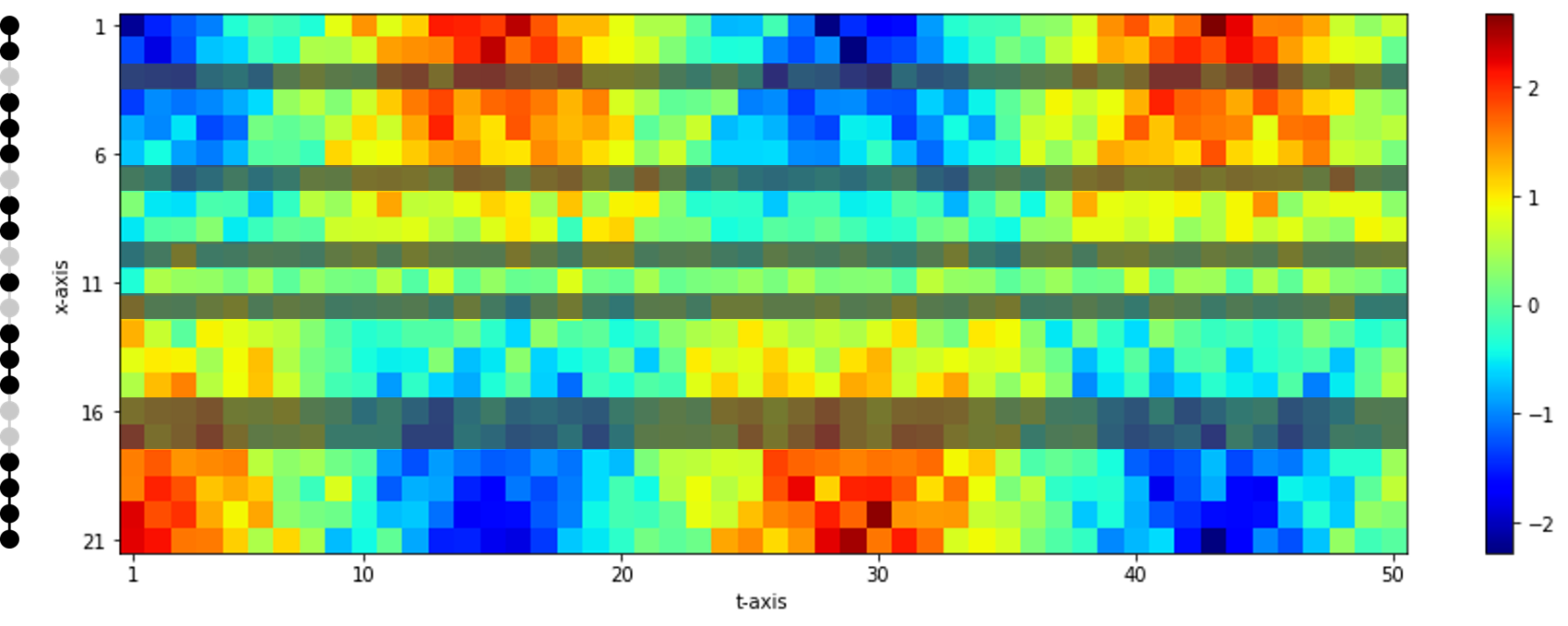}
\caption{Six observation points along the $x$-axis are missing from Figure \ref{fig:graph_wave}.
The expected graph on 15 black points is shown on the left-hand side.}
\label{fig:graph_missing}
\end{center}
\end{figure}
For $V = \{1\leq \x \leq 21\}$, assume the values on $\x = 3, 7, 10, 12, 16$, and $17$ are missing.
Then, a recovered graph is supposed to be the subgraph of the path graph in Figure \ref{fig:graph_missing}.

To this dataset, we apply the same procedure as above, namely, take $L = 3$ and $N = 26$, and then obtain $53$ modes $\{\hat{z}_j\}$.
As the result on the left-hand side of Figure \ref{fig:graph_recover_missing} shows, the extracted modes are basically the same as in the previous case, i.e., Figure \ref{fig:graph_recover}.
Then, the corresponding amplitudes $\{\hat{\alpha}_j\}$ and weights $\{\hat{\ww}_{\x\y}\}$ are also estimated, as on the right-hand side in Figure \ref{fig:graph_recover_missing}, which is close to what we expected in Figure \ref{fig:graph_missing}.
\begin{figure}[tp]
\begin{center}
\includegraphics[height=4.5cm]{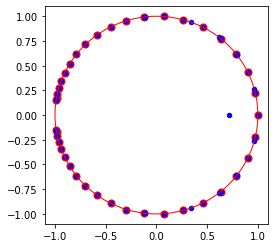}
\ 
\includegraphics[height=4.5cm]{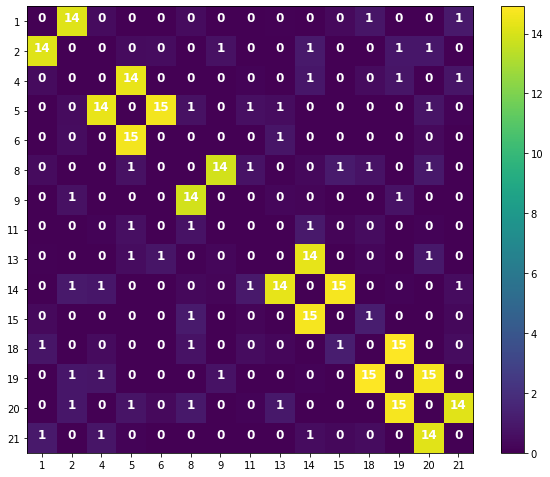}
\caption{[left] Forty-one true modes (red) and 53 estimated modes (blue).
[right] Recovered graph structure (10x), which is same as the subgraph of the path graph in \eqref{fig:graph_missing} up to a constant multiplication.}
\end{center}
\label{fig:graph_recover_missing}
\end{figure}

These examples show that Algorithm \ref{algo:algorithm} can correctly recover the underlying graph and subgraph from an observed graph wave function as Theorem \ref{thm:weight_recover} guarantees.
In the following subsections, we consider other cases in which multivariate signals are not explicitly assumed to have a graph structure under variables.

\subsection{Continuous wave over an interval}
\label{subsec:cont_wave_demo}

Unlike in \S \ref{subsec:graph_wave_demo}, here we consider a continuous wave function on the interval explained in \eqref{eq:cont_wave_function}.
Then, we apply Algorithm \ref{algo:algorithm} to its observed values in expectation that the algorithm extracts a graph in the shape of an interval, although they do not satisfy the graph wave equation over a path graph.
\begin{figure}[t]
\begin{tabular}{cc}
\begin{minipage}[c]{0.45\hsize}
\begin{center}
\includegraphics[width=7cm]{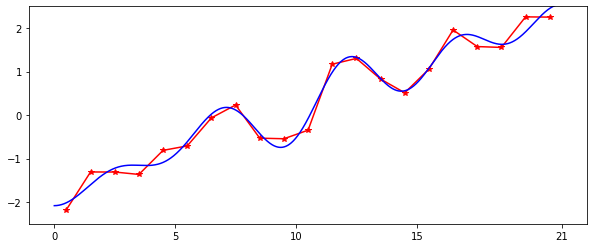}
\end{center}
\end{minipage}
&
\begin{minipage}[c]{0.45\hsize}
\begin{center} 
\includegraphics[height=5cm]{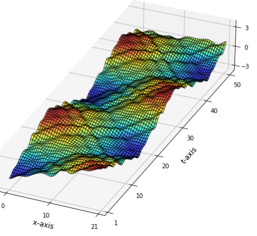}
\end{center}
\end{minipage}
\end{tabular}
\caption{[left] From $\tilde{f}$ (red), we define an initial function $\widetilde{W}(x,0)$ (blue) by taking the leading $11$ terms of the graph Fourier series.
[right] Continuous wave function given from the initial functions.}
\label{fig:cont_wave_demo}
\end{figure}

Take $n=21$ in \eqref{eq:cont_wave_function} as above.
For an initial function $W(x,0)$, let us use a smoothened function of $\tilde{f}$ in \eqref{eq:wave_init}.
By using the identities
\begin{align*}
a_0 = \int_0^nW(x,0)dx, \ a_k = 2\int_0^nW(x,0)\cos\frac{\pi k x}{n}dx \ \text{ for }k>0,
\end{align*}
we define
\begin{align*}
\tilde{a}_0 = \frac{1}{21}\sum_{\x\in V}\tilde{f}(\x), \ \tilde{a}_k = \frac{2}{21}\sum_{\x\in V}\tilde{f}(\x)\cos\frac{\pi k (2\x-1)}{42} \ \text{ for }k>0.
\end{align*}
In fact, these coefficients are nothing but the graph Fourier coefficients of $\tilde{f}$, namely, $\sqrt{21/2}\tilde{a}_k = \CF[\tilde{f}]_k$ for $1 \leq k \leq n-1$.
Then, by the taking first $11$ coefficients and assuming $\frac{\d}{\d t}|_{t=0}\widetilde{W}(x,t) = 0$, i.e., $b_k=0$ for any $k$, we obtain a continuous wave function
\begin{align}
\widetilde{W}(x,t) = \tilde{a}_0 + \sum_{k=1}^{10} \tilde{a}_k \cos\frac{\sqrt{c}\pi kt}{21}\cos\frac{\pi kx}{21}
\label{eq:cont_wave_function_demo}
\end{align}
for $\sqrt{c}=1.5$, which is shown in Figure \ref{fig:cont_wave_demo}.
Here, we emphasize that this wave function is continuously defined for $x \in [0, 21]$ and $t \in \R$, and thus there exist various ways to discretize it by evaluation.
Moreover, the number $11$ is taken so that each frequency $\sqrt{c}\pi k/21$ satisfies Nyquist condition \eqref{eq:Nyquist}, although this has another purpose, as explained below.

First of all, we evaluate the above $\widetilde{W}(x,t)$ at $21$ points along the $x$-axis, $x = \x-0.5$ for $\x = 1, 2, \cdots, 21$, and integer points along the time axis.
\begin{figure}[t]
\begin{center}
\includegraphics[height=4.5cm]{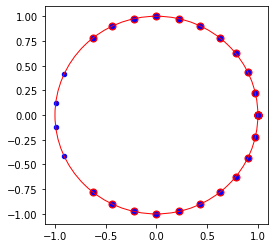}
\ 
\includegraphics[height=4.5cm]{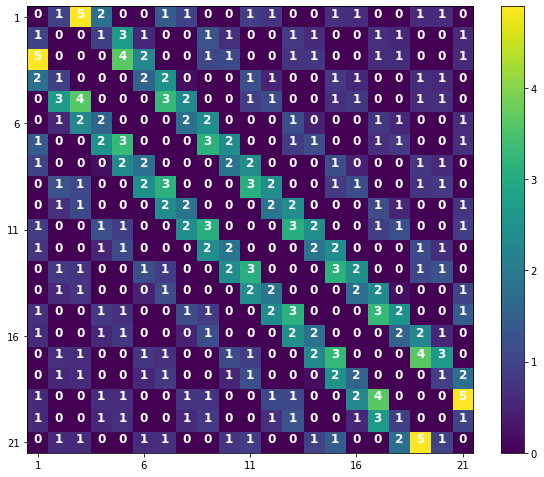}
\caption{[left] Twenty-one true modes (red) and 25 estimated modes (blue).
[right] Estimated graph structure (10x) over the 21 observation points.}
\label{fig:real_graph_recover}
\end{center}
\end{figure}
For these observed values, we take $L = 1$ and $N = 12$ in Step 1.
Then, we obtain $25$ modes $\{\hat{z}_j\}$ from $26$ observation points from Steps 2 and 3, as on the left-hand side in Figure \ref{fig:real_graph_recover}, and the corresponding amplitudes $\{\hat{\alpha}_j\}$ and weights $\{\hat{\ww}_{\x\y}\}$ from Steps 4 and 6, as on the right-hand side.
\begin{figure}[t]
\begin{center}
\includegraphics[height=3cm]{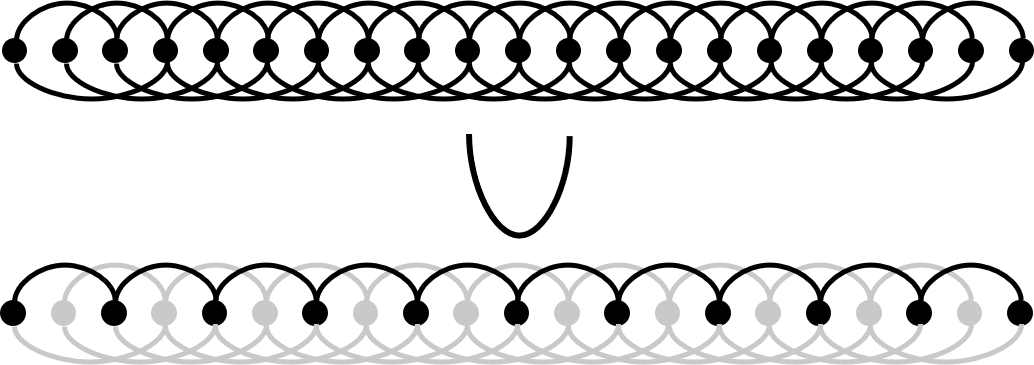}
\caption{Upper estimated graph having the lower path graph as a subgraph.}
\label{fig:estimated_graph}
\end{center}
\end{figure}

The estimated weights indicate the path-like graph at the top of Figure \ref{fig:estimated_graph}, which includes the path graph below.
Thus, the same procedure for a partial observed dataset like in Figure \ref{fig:graph_missing} is applicable to obtain the path graph.
\begin{figure}[tp]
\begin{center}
\includegraphics[height=3cm]{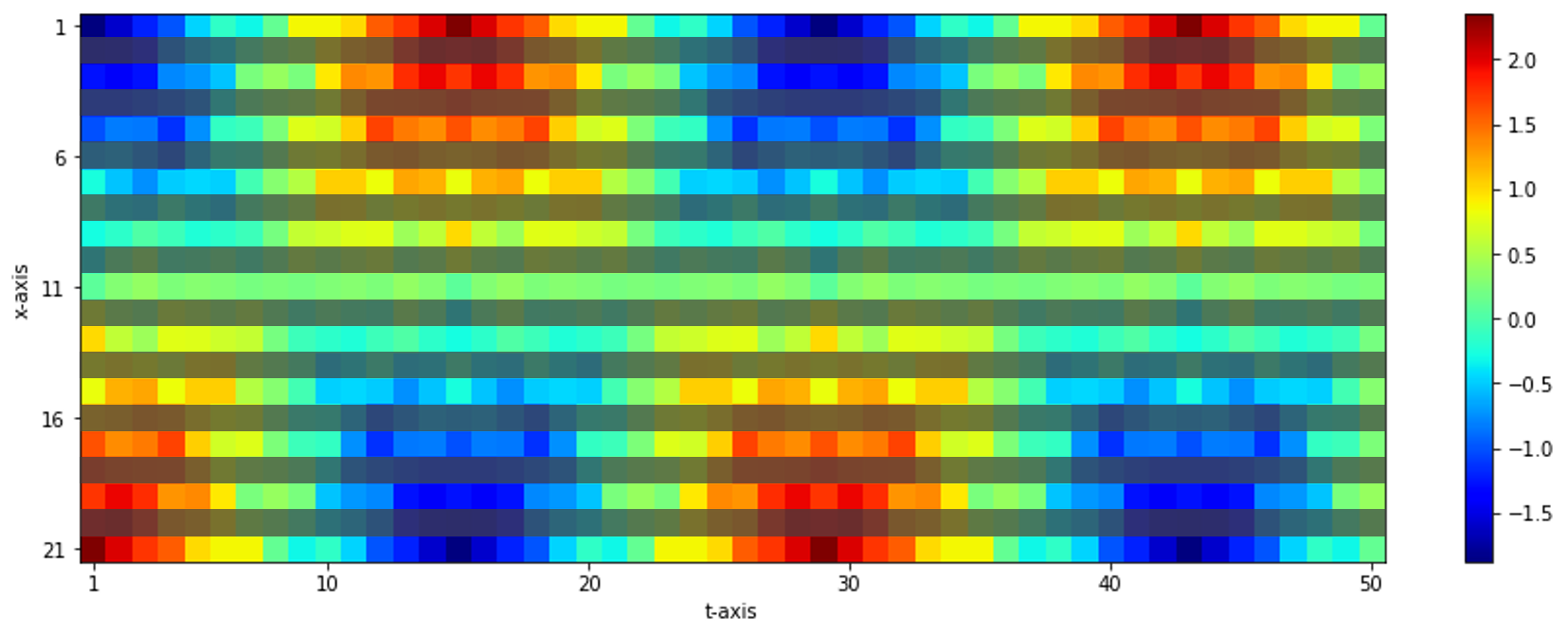}
\
\includegraphics[height=3cm]{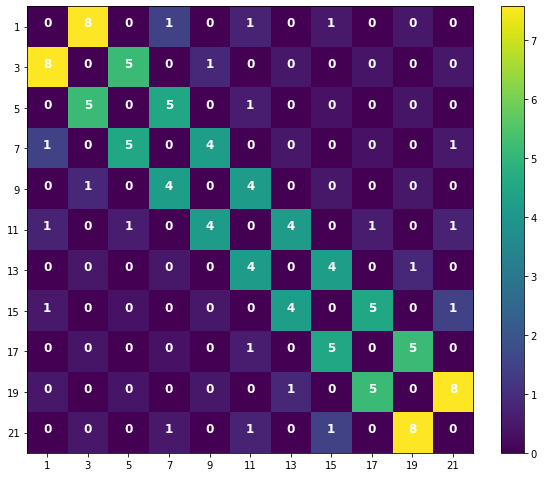}
\caption{[left] Eleven observation points reduced from the 21 points.
[right] Estimated graph structure (10x) over the 11 observation points.}
\label{fig:real_graph_recover_missing}
\end{center}
\end{figure}
Namely, we try another discretization $x = 2\x-1.5$ for $\x = 1, 2, \cdots, 11$, given on the left-hand side in Figure \ref{fig:real_graph_recover_missing}.
Then, by using the same parameter, we successfully obtain the weights close to those of the path graph on the right-hand side in Figure \ref{fig:real_graph_recover_missing} expected from the bottom path graph in Figure \ref{fig:estimated_graph}.

This phenomenon can be explained in another way.
From $21\times 26$ observed values $\{ \widetilde{W}(x,t) \in \R \mid x = 0.5, 1.5, \cdots, 20.5, t = 1, 2, \cdots, 26 \}$, a $(21,12)$ matrix formed by $\{G(\x,t;j)\}$ is constructed as in linear equation \eqref{eq:real_const_dmd_matrix} in Step 2.
\begin{figure}[tp]
\begin{center}
\includegraphics[height=3cm]{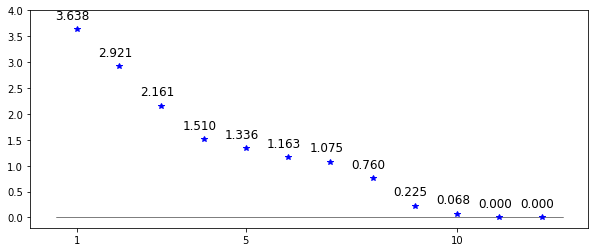}
\caption{Twelve singular values of the $(21,12)$ matrix.}
\label{fig:singular_values}
\end{center}
\end{figure}
Here, we can see that the $(21,12)$ matrix has rank $10$, which means its singular values consist of $10$ major values and $2$ minor values, as in Figure \ref{fig:singular_values}.
This is consistent with the definition of $\widetilde{W}(x,t)$, where we take the $10$ terms in the summation other than the constant term \eqref{eq:cont_wave_function_demo}.
Hence, it is sufficient to use more than $10$ observation points, such as $11\times 26$ observed values $\{ \widetilde{W}(x,t) \in \R \mid x = 0.5, 2.5, \cdots, 20.5, t = 1, 2, \cdots, 26 \}$ as above, in order to extract an outline of the underlying graph.

\subsection{Pose model}
\label{subsec:pose_model}
As a more practical case, we consider a graph estimation problem from human joint tracking data.
To this end, we use an open source dataset PKU-MMD, which consists of RGB, skeletons, depth, and IR sequential data for 51 action categories \cite{liu2017pku}.
The skeleton data we use herein are obtained by Kinect v2, the joints of which are described as in Figure \ref{fig:skeleton} \cite{kinectjoint}.
More strictly speaking, we use the 2,100th to 2,800th frames from the "0007-M" sample in the dataset.
\begin{figure}[tp]
\begin{tabular}{cc}
\begin{minipage}[c]{0.45\hsize}
\begin{center}
\includegraphics[height=5cm]{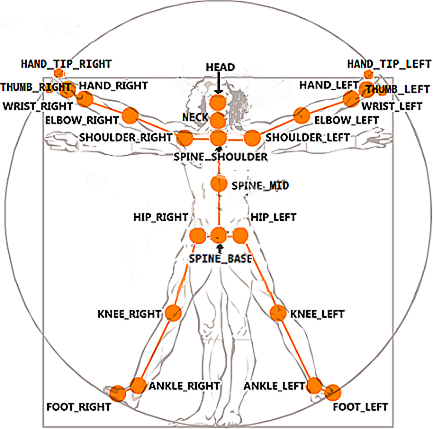}
\end{center}
\end{minipage}
&
\begin{minipage}[c]{0.45\hsize}
\begin{center}
\includegraphics[height=4.5cm]{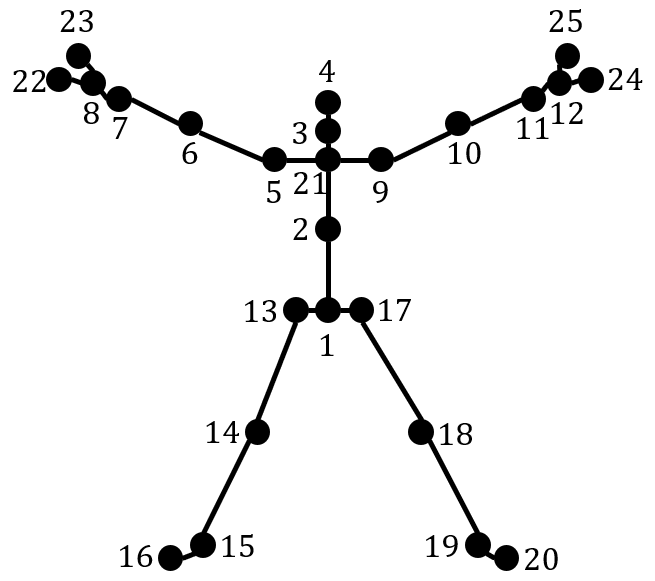}
\end{center}
\end{minipage}
\end{tabular}
\caption{[left] Meanings of 25 joints.
[right] Corresponding joint numbers.}
\label{fig:skeleton}
\end{figure}
Although the skeleton data provide the three-dimensional coordinates $x$, $y$, and $z$ of each joint location in space, we only use the $x$ and $y$ coordinates in this experiment.
Then, we try to extract relations among $25$ moving $x+\i y$ joint locations as graphs, which are supposed to characterize individual actions.
\begin{figure}[p]
\begin{center}
\includegraphics[width=15cm]{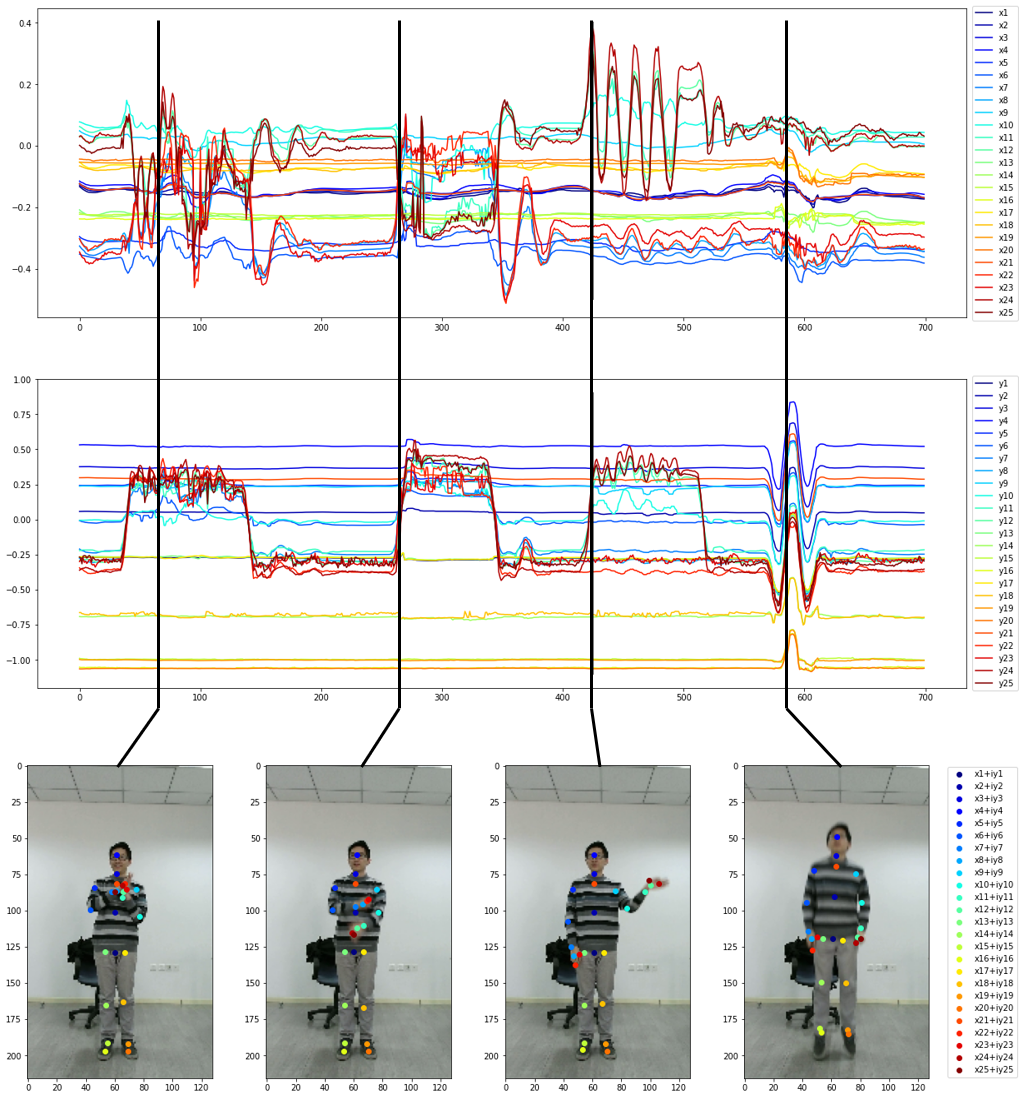}
\caption{[top] Time-series data of the 25 $x$-coordinates of human joints.
[middle] Time-series data of the 25 $y$-coordinates of human joints.
[bottom] Four actions starting from $t=65, 265, 425$, and $585$ with the 25 $x+\i y$-coordinates of human joints.}
\label{fig:data_description}
\end{center}
\end{figure}

Let $F_x(\x,t)$ and $F_y(\x,t)$ denote functions of the $x$ and $y$ locations in an RGB image, respectively, for the joint nodes $1 \leq \x \leq 25$ and the times $t = 0, 1, \cdots, 600$.
First, as $n=2\times25=50$ time series, we apply Steps 1 to 5 of Algorithm \ref{algo:algorithm} for $L=4$ and $N=5$ to a $15$-long sequence starting from the $t_1$-th frame.
Then, we obtain mode decompositions
\begin{align}
\hat{F}_x(\x,t; t_1)=\hat{\alpha}_{x,0,t_1}(\x) + \sum_{j=1}^{10}\hat{\alpha}_{x,j,t_1}(\x)\hat{z}_j^t, \ \ \hat{F}_y(\x,t; t_1)=\hat{\alpha}_{y,0,t_1}(\x) + \sum_{j=1}^{10}\hat{\alpha}_{y,j,t_1}(\x)\hat{z}_j^t
\label{eq:pose_decompose}
\end{align}
for any node $\x$ so as to coincide with $F_x(\x,t)$ and $F_y(\x,t)$ on the sequence $t_1 \leq t < t_1+15$.
See also Lemma \ref{lem:complex_amp}.
Here, we can see that $\hat{F}_x(\x,t;t_1)+\i \hat{F}_y(\x,t;t_1)$ represents the time evolution of the $x+\i y$ location in an RGB image.
Finally, we calculate \eqref{eq:weight_estimate} for $\hat{\alpha}_{j,t_1} = \hat{\alpha}_{x,j,t_1}+\i \hat{\alpha}_{y,j,t_1}$ in order to estimate weights among joint locations at $t=t_1$.
Note that although $F_x(\x,t)$ and $F_y(\x,t)$ no longer satisfy the graph wave equation for the estimated graph, we suppose that the graph weights extract meaningful features from the decomposition \eqref{eq:pose_decompose}. 

For example, we pick out four sequences representing "clapping" from the 2,165th to 2,180th frames ($t=65$), "cross hands in front" from the 2,365th to 2,380th frames ($t=265$), "hand waving" from the 2,525th to 2,540th frames ($t=425$), and "jump up" from the 2,685th to 2,700th frames ($t=585$).
The corresponding data are shown in Figure \ref{fig:data_description}.
Then, we draw the resulting graphs in Figure \ref{fig:result_graph}, where each edge indicates that the estimated weight is larger than $0.2$.
\begin{figure}[tp]
\begin{center}
\includegraphics[width=15cm]{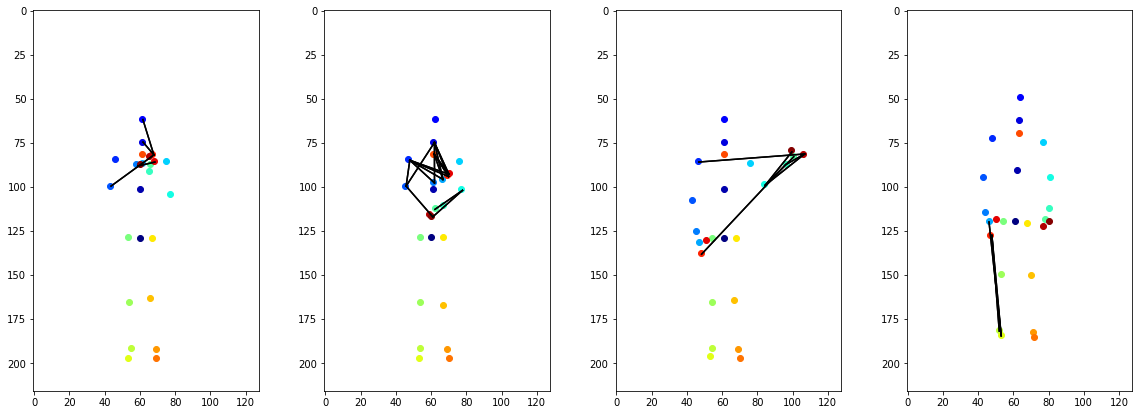}
\caption{Graphs of the results. From left to right, "clapping", "cross hands in front", "hand waving", and "jump up".}
\label{fig:result_graph}
\end{center}
\end{figure}
As we can see, edges in each case are illustrated on nodes and pairs that characterize each action.
In particular, since we are paying attention to positive weights, Algorithm \ref{algo:algorithm} ignores synchronized pairs and extracts anti-synchronized pairs.
Hence, even in the "jump up" case, only the leg motion is highlighted against all upward moving nodes.
As in this case, we expect that the proposed graph recovery algorithm will work well, even for non-wave signals.

\section{Conclusion}
In the present paper, we introduced a graph recovery method from a graph wave function based on a modified DMD algorithm.
Then, we showed its effectiveness by using three examples, i.e., a graph wave function, a continuous wave function, and human joint tracking data, from viewpoints of from theoretical demonstration to practical study.
In the present study, we do not assume the existence of noise in each time-series model.
Hence, further study is needed in order to assure that the proposed algorithm works against noises, such as white noise, with the help of known technologies in signal processing, time-series analysis, and statistics.

\section*{Acknowledgements}
The author would like to thank Yosuke Otsubo for discussing and suggesting DMD technology for further study.
The author is also grateful to Satoshi Takahashi, Tetsuya Koike, Chikara Nakamura, Bausan Yuan, Ping-Wei Chang, and Pranav Gundewar for their constant encouragement.

\small

\bibliographystyle{plain}


\end{document}